\documentclass[11pt, a4paper]{article}
\usepackage{amsmath, amsfonts, amssymb, amsthm,mathtools}
\usepackage{enumitem}
\usepackage{cite}
\usepackage{graphicx}
\usepackage{color}
\usepackage[top=1in, bottom=1in, left=1in, right=1in]{geometry}
\usepackage{comment}
\usepackage{url}
\usepackage{booktabs}
\usepackage{array}
\newtheorem{theorem}{Theorem}[section]
\newtheorem{lemma}[theorem]{Lemma}
\newtheorem{corollary}[theorem]{Corollary}
\newtheorem{proposition}[theorem]{Proposition}
\newtheorem{example}[theorem]{Example}
\newtheorem{remark}[theorem]{Remark}
\newtheorem{definition}[theorem]{Definition}

\newcommand{\F}{\mathbb{F}}
\newcommand{\R}{\mathcal{R}}
\newcommand{\K}{\mathbb{K}}
\newcommand{\eps}{\varepsilon}
\newcommand{\ip}[2]{\langle #1,#2\rangle}
\newcommand{\ipf}[2]{\langle #1,#2\rangle_{f}}

\newcommand{\wt}{\operatorname{wt}}
\newcommand{\swt}{\operatorname{swt}}

\newcommand{\dpe}{\perp_{E}}
\DeclareMathOperator{\Hom}{Hom}

\title{Annihilator and twisted Euclidean duality for quasi-polycyclic codes}

\author{Tushar Bag$^{1}$,  Edgar Martínez-Moro$^{2}$, Daniel Panario$^3$}

\date{
\small{
1. Quantum Technology Institute \& Department of Mathematics, SRM University-AP,\\ Amaravati 522240, Andhra Pradesh, India \\
2.  IMUVa-Institute of Mathematics, University of Valladolid, Spain\\
3. School of Mathematics and Statistics, Carleton University, Ottawa, Ontario K1S 5B6, Canada.
}
\today}

\begin{document}

\maketitle

\begingroup
\renewcommand\thefootnote{}
\footnotetext{
Email: tusharbag2011@gmail.com (T. Bag) [corresponding author], edgar.martinez@uva.es (E. Mart\'inez-Moro),
daniel@math.carleton.ca (D. Panario)
}
\endgroup

\begin{abstract}
Let $f\in\F_q[x]$ be a monic polynomial of degree $m$ with $f(0)\ne0$,
and let $\R=\F_q[x]/\langle f\rangle$. Under coefficient expansion, a
quasi-polycyclic (QP) code of index $n$ corresponds to an $\R$-submodule
of $\R^n$. In this paper, we study QP codes with respect to the
annihilator duality.  
We show that this form is non-degenerate and that the
annihilator dual of a QP code is again a QP code. We also give an
equivalent description of the dual in terms of the $\R$-valued dot
product, which leads to self-orthogonality criteria.
We determine the Gram matrix of the annihilator form and obtain an
explicit formula for its determinant. In coefficient coordinates, this
shows that the annihilator dual can be viewed as a twisted Euclidean
dual. Using this description, we characterize when a coordinatewise
$\F_q$-linear map converts annihilator duality into ordinary Euclidean
duality.
For squarefree $f$, we show that annihilator duality decomposes into
ordinary Euclidean duality on the components arising from the Chinese
Remainder Theorem. This gives simple criteria for self-orthogonal,
self-dual, dual-containing, and complementary-dual QP codes. We show how the annihilator dual interacts with the Hamming weight enumerator and compute the MacWilliams transform associated with that duality.  Finally, we apply these results to Calderbank--Shor--Steane
 and Steane-enlarged quantum-code constructions over $\R$ and, when
a suitable duality-preserving coordinate map exists, over $\F_q$. This gives binary and ternary stabilizer codes with minimum-distance lower bounds matching the best known bounds, most of which arise from rings $\R$ that are not fields.

\end{abstract}



\section{Introduction}
A central theme in algebraic coding theory is the interaction between a
code and its dual. Many of the most useful families of codes are defined
by an invariance property. Cyclic codes are invariant under the cyclic
shift, and constacyclic codes under a twisted shift. One then asks
whether the dual code is invariant in the same way, so that the family is
closed under duality. When it is, self-orthogonal codes, linear
complementary dual (LCD) codes and quantum codes can all be studied
within the same family.

\emph{Polycyclic codes} provide a  common generalization of those families. Fix a
finite field $\F_q$, where $q$ is a power of a prime; a linear code of length $m$ and dimension $k$ is just a $ k$-dimensional subspace of $\mathbb F_q^m$. If we fix a monic polynomial $f\in\F_q[x]$ of degree $m$, a
polycyclic code associated with $f$ is a linear code closed
under the \emph{$f$-polycyclic shift}, which moves each coordinate one
step to the right and feeds the overflow back according to the
coefficients of $f$.

Under the usual coefficient identification between $\F_q^m$ and
$\R=\F_q[x]/\langle f\rangle$, an $f$-polycyclic code corresponds to an
ideal of $\R$. We shall use these two equivalent descriptions
interchangeably when no confusion can arise.
Taking $f=x^m-1$ recovers cyclic codes and
$f=x^m-\lambda$ recovers constacyclic codes. The name \emph{polycyclic}
was used by L\'opez-Permouth et al. \cite{LPS}. Berger and El Amrani~\cite{BEA} studied codes over finite quotients of
polynomial rings; in the present notation, a quasi-polycyclic (QP) code
of index $n$ corresponds, under coordinatewise coefficient expansion,
to an $\R$-submodule of $\R^n$, so that the same $f$-polycyclic shift
acts simultaneously on the $n$ blocks.

The difficulty appears when one considers duality. Alahmadi et al. \cite{ADLS}
observed that the Euclidean dual of a polycyclic code is in general
\emph{not} polycyclic: it is closed under the adjoint (``sequential'')
shift instead, so it leaves the domain in which one is working. 
Their solution was to change the bilinear form used to define the dual
rather than the code itself. They introduced the annihilator form
$\langle g,h\rangle_f=\eps(gh),$
where $\eps:\R\to\F_q$ returns the constant term of the reduction of
$gh$ modulo $f$. The dual taken with respect to this form, called the
annihilator dual and written $C^\circ$, is again polycyclic.
Fotue-Tabue et al. \cite{FMB} developed this duality over finite chain rings, and in \cite{BMS,BMSSY} this duality was related to the Mattson-Solomon transform domain. Wu et al. \cite{WSS} considered
one-generator QP codes over finite chain rings, and Ou-azzou et al.
\cite{ONA} studied the algebraic structure of QP codes together with
quantum-code constructions related to them. Related recent work includes QP and skew
QP codes over $\F_q$ by Bag and Panario \cite{BP}, polycyclic codes over
serial rings and their annihilator CSS construction by Bajalan and
Mart\'inez-Moro \cite{BM}, and multi-generator generalized QP codes and
their constituent decompositions by Suxena et al. \cite{SPD}.
\vskip 3pt

Our starting point is the bilinear form $(g,h)\longmapsto \eps(gh)$
on $\R$. Under the standing assumption $f(0)\ne0$, this form is
non-degenerate. 
For $u=(u_1,\ldots,u_n)$ and $v=(v_1,\ldots,v_n)$ in $\R^n$,  if we consider the
$\R$-valued dot product in $\R^n$ given by
$
u\cdot v=\sum_{i=1}^n u_i v_i$, 
since $u\cdot v\in\R$, the constant-term map $\varepsilon$ can be applied
to it, giving the annihilator form
$\langle u,v\rangle=\varepsilon(u\cdot v).$
For a QP code $C\subseteq\R^n$, the $\R$-module structure together with
the non-degeneracy of the form shows that
$\varepsilon(u\cdot v)=0~\text{for all }u\in C$
is equivalent to
$u\cdot v=0~\text{for all }u\in C.$
Thus  we  denote by $ C^\circ=C^{\perp_{\mathcal R}}$, the  Euclidean orthogonal over $\mathcal R$, for an arbitrary index and any number of generators, and therefore $C^\circ$ is again a QP code.  This is consistent with the  index 1 case, where $C^\circ=C^{\perp_{\mathcal R}}=\text{Ann}_{\mathcal R}(\langle g\rangle)$,  for a
code generated by $g$; see \cite{BMS} for a proof.   We note in
Remark~\ref{rem:wss} that the proof of \cite[Prop.~4.1]{WSS} appears to
use a scalar-extraction identity that does not hold for a general
multiplier; the argument above avoids that identity.

A second objective is to make the annihilator form explicit in
coordinates. Its Gram matrix $G$ is a Hankel matrix, and we prove that its determinant is given by
$
\det G=(-1)^{(m-1)(m-2)/2}f_0^{\,m-1}.
$
The coordinate description that we provide also shows that the annihilator dual can be
viewed as a twisted Euclidean dual in $\mathbb F_q^{mn}$. The determinant gives the
discriminant needed to determine when a coordinatewise change of basis
converts annihilator orthogonality into Euclidean orthogonality.
\vskip 3pt
The matrix problem that arises here is closely related to classical
factorizations of symmetric matrices over finite fields, as in Seroussi
and Lempel \cite{SL}. Duality-preserving maps from codes over
rings to codes over smaller rings were later studied through suitable
bases by Szabo and Ulmer \cite{SU}. In this paper, a
\emph{coordinatewise duality-preserving Gray map} means a map induced by
one matrix $S\in \mathrm{GL}_m(\F_q)$ on each $\R$-coordinate.  No Hamming
isometry is included in the definition.  We prove that such a map exists
if and only if $G$ is congruent to a scalar multiple of $I_m$ (the $m\times m$ identity matrix).  Thus it
always exists when $q$ is even or $m$ is odd; for $q$ odd and $m$ even,
it exists exactly when $\det G$ is a square in $\F_q^\ast$.

When $f$ is squarefree, the Chinese Remainder decomposition turns
$C^\circ$ into the direct sum of the ordinary Euclidean duals of the
constituents.  The identity $C^\circ=C^{\perp_{\mathcal R}}$  also gives the generator
criterion $MM^{\mathsf T}=0$ for annihilator self-orthogonality and, for a one-generator code generated by   $g=(g_1,\dots,g_n)$ it gives
$\sum_{i=1}^n g_i(x)^2\equiv0\pmod f.$
We also apply these duality results to CSS quantum-code constructions.
For background on stabilizer codes over finite fields, we refer to
\cite{CRSS,KKKS}.

The paper is organized as follows.  Section~\ref{sec:prelim} fixes the
module and CRT notation.  Section~\ref{sec:form} develops the annihilator form and proves $C^\circ=C^{\perp_{\mathcal R}}$.
Section~\ref{sec:gram} computes the Gram matrix and its determinant and describes the
annihilator dual as a twisted Euclidean dual.
Section~\ref{sec:gray} classifies coordinatewise duality-preserving
maps.  Section~\ref{sec:duals} gives the syzygy, CRT, and
self-orthogonality consequences.  Section~\ref{sec:mw} presents the
corresponding MacWilliams transform. Section~\ref{sec:css} gives the
quantum-code constructions, and Section~\ref{sec:conclusion} concludes the paper.

\section{Preliminaries}\label{sec:prelim}
Let $\F_q$ be the finite field with $q=p^t$ elements ($p$ a prime number) and fix a monic
polynomial $f\in\mathbb F_q[x]$ of degree $m\ge1$,
\begin{equation}\label{eq:f}
f(x)=x^{m}-\bigl(f_{m-1}x^{m-1}+\cdots+f_1x+f_0\bigr)\in\F_q[x],
\qquad f_0\ne0 .
\end{equation}
We  assume  $f_0\ne0$ (equivalently, $f(0)\ne0$)   throughout  the paper.
We consider the principal ideal ring $\R=\F_q[x]/\langle f\rangle$. Every element of $\R$ has a unique representative of degree smaller than $m$, so $\R$
is  an $\F_q$-algebra of dimension $m$ with basis
$\mathcal B=\{1,x,\dots,x^{m-1}\}$.

In $\R$ the relation
\eqref{eq:f} reads
\begin{equation}\label{eq:xm}
x^{m}=f_{m-1}x^{m-1}+\cdots+f_1x+f_0 ,
\end{equation}
which is the rule used to reduce a product of two representatives of a class in $\R$ back to a representative of  degree smaller than $m$. The following elementary observation is the only place where $f_0\ne0$ is
used directly, and it will be needed repeatedly.

\begin{lemma}\label{lem:xunit}
The class of $x$ is a unit in $\R$, with
\[
x^{-1}=f_0^{-1}\bigl(x^{m-1}-f_{m-1}x^{m-2}-\cdots-f_1\bigr).
\]
\end{lemma}

\begin{proof}
We write $g(x)=f_0^{-1}(x^{m-1}-f_{m-1}x^{m-2}-\cdots-f_1)$; this is
well-defined because $f_0\ne0$. Multiplying by $x$ gives
$
x\,g(x)=f_0^{-1}\bigl(x^{m}-f_{m-1}x^{m-1}-\cdots-f_1x\bigr)$, and
substituting \eqref{eq:xm} for $x^m$, the terms
$f_{m-1}x^{m-1},\dots,f_1x$ cancel and only $f_0$ survives, so
$x\,g(x)=f_0^{-1}f_0=1$.
\end{proof}

Let $\psi:\R\to\F_q^m$ be the coefficient expansion with respect to the basis
$\mathcal B=\{1,x,\ldots,x^{m-1}\}$. Thus, if $a=\sum_{\ell=0}^{m-1}a_\ell x^\ell\in\R,$
then $\psi(a)=(a_0,\ldots,a_{m-1}).$
Applying $\psi$ coordinatewise, we obtain the map $\Psi:\R^n\longrightarrow\F_q^{mn}.$
More precisely, if $u=(u_1,\ldots,u_n)\in\R^n$ with $u_i=\sum_{\ell=0}^{m-1}u_{i,\ell}x^\ell,$
then $\Psi(u)=
(u_{1,0},\ldots,u_{1,m-1},\,
u_{2,0},\ldots,u_{2,m-1},\ldots,
u_{n,0},\ldots,u_{n,m-1}).$

\begin{definition}\label{def:QP}
An \emph{$f$-quasi-polycyclic (QP) code of index $n$} is an
$\F_q$-linear code of the form
$
\Psi(C)\subseteq\F_q^{mn},
$
where $C\subseteq\R^n$ is an $\R$-submodule.\end{definition}

Thus, via the coordinate
map $\Psi$, $f$-quasi-polycyclic codes are in one-to-one correspondence
with $\R$-submodules of $\R^n$. When no confusion  arises, we identify
these two descriptions and also refer to the submodule $C$ itself as a
QP code.
The map $\Psi$ is an isomorphism of $\F_q$-vector spaces. If we write
$|C|$ for the cardinality of $C$, then
$
|C|=q^{\dim_{\F_q}\Psi(C)}$.

In order to check explicitly why Definition~\ref{def:QP} matches the
description of QP codes by an $f$-polynomial shift, it is enough  to see that multiplication by $x$ on $\R$ is an
$\F_q$-linear map; in the basis $\mathcal B$ its  companion
matrix is
\[
T_f=
\begin{pmatrix}
0&0&\cdots&0&f_0\\
1&0&\cdots&0&f_1\\
0&1&\cdots&0&f_2\\
\vdots& &\ddots& &\vdots\\
0&0&\cdots&1&f_{m-1}
\end{pmatrix}.
\]
The matrix $T_f$ acts on coordinate columns, so that,
$\Psi(u)$ being a row vector, $\Psi(xu)=\Psi(u)T_f^{\mathsf T}$.  Thus
multiplication by $x$ on $\R^n$ corresponds, under $\Psi$, to
right multiplication by the block-diagonal matrix
$I_n\otimes T_f^{\mathsf T}$ on row vectors of
$\F_q^{mn}$: it applies the $f$-polycyclic shift to each of the $n$
blocks of length $m$ simultaneously.

\begin{proposition}
A subset $C\subseteq\R^n$ is an $\R$-submodule if and only if $\Psi(C)$ is
an $\F_q$-linear subspace of $\F_q^{mn}$ that is invariant under the
simultaneous $f$-polycyclic shift of its $n$ blocks.
\end{proposition}

\begin{proof}
Suppose first that $C$ is an $\R$-submodule. Then $C$ is closed under
addition and under multiplication by every element of $\R$. Since
$\F_q\subseteq\R$, it follows that $C$ is an $\F_q$-subspace of
$\R^n$.

The map $\Psi:\R^n\to\F_q^{mn}$ is an $\F_q$-linear isomorphism.
Therefore, if $u,v\in C$ and $a\in\F_q$, then
$\Psi(u+v)=\Psi(u)+\Psi(v)$ and
$\Psi(au)=a\Psi(u)$. Since $u+v,au\in C$, we obtain
$\Psi(u)+\Psi(v)\in\Psi(C)$ and
$a\Psi(u)\in\Psi(C)$. Hence $\Psi(C)$ is an $\F_q$-linear code of
length $mn$.

Also, since $x\in\R$ and $C$ is an $\R$-submodule, $xu\in C$ for every
$u\in C$. Under $\Psi$, multiplication by $x$ corresponds to the
simultaneous $f$-polycyclic shift of the $n$ blocks. Hence, whenever
$\Psi(u)\in\Psi(C)$, its simultaneous $f$-polycyclic shift
$\Psi(xu)$ also belongs to $\Psi(C)$. Therefore $\Psi(C)$ is invariant
under this shift.

Conversely, suppose that $\Psi(C)$ is an $\F_q$-linear code of length
$mn$ and is invariant under the simultaneous $f$-polycyclic shift.
Since $\Psi$ is an $\F_q$-linear isomorphism, $C$ is an
$\F_q$-subspace of $\R^n$. Let $u\in C$. Then
$\Psi(u)\in\Psi(C)$. By shift invariance, $\Psi(xu)\in\Psi(C)$, and
since $\Psi$ is one-to-one, $xu\in C$. Repeating this gives
$x^ju\in C$ for all $j\geq0$.

Now every $r\in\R$ can be written as
$r=a_0+a_1x+\cdots+a_{m-1}x^{m-1}$ with $a_i\in\F_q$. Hence
$ru=a_0u+a_1xu+\cdots+a_{m-1}x^{m-1}u\in C$, since $C$ is an
$\F_q$-subspace. Therefore $C$ is an $\R$-submodule.
\end{proof}

Thus, Definition~\ref{def:QP} agrees with the classical shift description, and for
$n=1$, the corresponding $\R$-submodules are precisely the ideals of
$\R$, giving the usual polycyclic codes.

\begin{example}{\em
Let $q=2$, $f(x)=x^2+x+1$, so $m=2$ and $\R=\F_2[x]/\langle
x^2+x+1\rangle\cong\F_4$. Take $n=2$ and let $C=\langle(1,x)\rangle$ be
the QP code generated by the single  polynomial vector $g=(1,x)$. As an
$\F_2$-space, $C$ is spanned by $\Psi(g)=(1,0\,|\,0,1)$ and
$\Psi(xg)=\Psi(x,x+1)=(0,1\,|\,1,1)$, so $\Psi(C)$ is a $[4,2]$ binary
code. One  can check that $\Psi(C)$ is invariant under the
simultaneous shift.}
\end{example}

It is well known that the polynomial $f$  can be factored into distinct monic irreducibles,
$f=\prod_{j=1}^{k}p_j^{e_j}$. Since the $p_j^{e_j}$ are pairwise
coprime, the Chinese Remainder Theorem gives an isomorphism of
$\F_q$-algebras
\begin{equation}\label{eq:crt}
\R\;\cong\;\bigoplus_{j=1}^{k}\R_j,\qquad
\R_j=\F_q[x]/\langle p_j^{e_j}\rangle,\quad j=1,\ldots , k.
\end{equation}
Each  component $\R_j$  for $j=1,\ldots , k$ is a finite chain ring: its ideals form the chain
$\R_j\supset\langle p_j\rangle\supset\langle
p_j^{2}\rangle\supset\cdots\supset\langle p_j^{e_j}\rangle=\{0\}$.  Consequently $\R$ is a finite
direct product of finite chain rings and hence a finite commutative
Frobenius ring; see also the serial-ring setting of \cite{BM}.
When $f$ is a squarefree polynomial, i.e.\ $e_j=1$ for all $j$, each factor
$\R_j=\K_j:=\F_q[x]/\langle p_j\rangle$ is a \emph{field}, namely
$\F_{q^{\deg p_j}}$, and $\R$ is semisimple. Applying \eqref{eq:crt}
coordinatewise, an $\R$-submodule $C\subseteq\R^n$ decomposes as
\begin{equation}\label{eq:Cdecomp}
C\;\cong\;\bigoplus_{j=1}^{k}C_j,\qquad C_j\subseteq\R_j^{\,n}
\ \text{an $\R_j$-linear code}, \quad j=1,\ldots , k,
\end{equation}
exactly as in \cite[Theorem~1]{BM}. We call $C_j$ the \emph{$j$-th
constituent} of $C$. The decomposition \eqref{eq:Cdecomp} will be the main
computational tool in Section~\ref{sec:duals}.

\begin{remark}
Definition~\ref{def:QP} is the module description used by Berger--El
Amrani \cite{BEA} and coincides with the QP codes of Bag--Panario
\cite{BP} and, over chain rings, with those of Wu--Shi--Sol\'e
\cite{WSS}. The polycyclic case $n=1$ is the setting of \cite{ADLS,BM,AS}.
\end{remark}

\section{The annihilator form on $\R^n$}\label{sec:form}
In this section we introduce the annihilator form used throughout the paper.  As recalled in
the introduction, the  $\mathbb F_q$-Euclidean dual of a QP code
need not be QP, so we replace the  $\mathbb F_q$-Euclidean form by the the $\mathcal R$-Euclidean form described in terms of the  annihilator.

Let $\eps:\R\to\F_q$ be the $\F_q$-linear functional returning the
constant coefficient of the representative of degree smaller than $m$.  For
$g,h\in\R$ put $\ipf{g}{h}=\eps(gh)$, as in \cite{ADLS}.  The
associativity of multiplication in $\R$ is the basic structural property
used below.

\begin{definition}\label{def:form}
The \emph{annihilator form} on $\R^n$ is
\[
\ip{u}{v}=\eps\!\Bigl(\sum_{i=1}^{n}u_iv_i\Bigr)=\eps(u\cdot v)
=\sum_{i=1}^{n}\ipf{u_i}{v_i}\in\F_q,
\qquad u,v\in\R^n ,
\]
where $u\cdot v=\sum_i u_iv_i\in\R$ denotes the $\R$-valued dot product.
The \emph{annihilator dual} of $C\subseteq\R^n$ is
\[
C^{\circ}=\{v\in\R^n:\ip{u}{v}=0\ \text{for all }u\in C\},
\]
and the \emph{$\mathcal R$}-Euclidean orthogonal of $C$ is
\[
C^{\perp_{\mathcal R}}=\{v\in\R^n:u\cdot v=0\ \text{in}\ \R\ \text{for all }u\in C\}.
\]
\end{definition}
We observe that $C^\circ$ requires only that the constant term of $u\cdot v$ vanish, whereas membership in $C^{\perp_{\mathcal R}}$ requires the entire element $u\cdot v\in\mathcal R$ to vanish.

\begin{lemma}\label{lem:bilinear}
The annihilator form is a symmetric $\F_q$-bilinear form on $\R^n$.
\end{lemma}

\begin{proof}
Since $\R$ is commutative, $u_iv_i=v_iu_i$ for each $i$, hence
$u\cdot v=v\cdot u$ and therefore $\ip{u}{v}=\eps(u\cdot
v)=\eps(v\cdot u)=\ip{v}{u}$; the form is symmetric.
For bilinearity, let $a,b\in\F_q$ and $u,w,v\in\R^n$. Using that the dot
product is $\F_q$-bilinear and that $\eps$ is $\F_q$-linear,
\[
\ip{au+bw}{v}
=\eps\bigl((au+bw)\cdot v\bigr)
=\eps\bigl(a(u\cdot v)+b(w\cdot v)\bigr)
=a\,\eps(u\cdot v)+b\,\eps(w\cdot v)
=a\ip{u}{v}+b\ip{w}{v}.
\]
Linearity in the second argument follows from the symmetry.
\end{proof}


\begin{lemma}\label{lem:selfadj}
For all $r\in\R$ and all $u,v\in\R^n$,
$
\ip{ru}{v}=\ip{u}{rv}.
$ That is, multiplication by any ring element is self-adjoint for the
annihilator form.
\end{lemma}

\begin{proof}
It is easy to check that
\[
(ru)\cdot v=\sum_{i=1}^{n}(ru_i)v_i=r\sum_{i=1}^{n}u_iv_i=r(u\cdot v)
=\sum_{i=1}^{n}u_i(rv_i)=u\cdot(rv).
\]
Applying $\eps$ to the two ends gives $\ip{ru}{v}=\eps\bigl((ru)\cdot
v\bigr)=\eps\bigl(u\cdot(rv)\bigr)=\ip{u}{rv}$.
\end{proof}


\begin{proposition}\label{prop:nondeg}
The annihilator form is non-degenerate on $\R^n$: if $u\in\R^n$
satisfies $\ip{u}{v}=0$ for all $v\in\R^n$, then $u=0$.
\end{proposition}

\begin{proof}
First take $n=1$.  Let
$w=\sum_{\ell=0}^{m-1}w_\ell x^\ell\ne0$, and let $j$ be the least
index for which $w_j\ne0$.  By Lemma~\ref{lem:xunit}, $x^{-j}\in\R$.
Since $w_0=\cdots=w_{j-1}=0$,
\[
   x^{-j}w=w_j+w_{j+1}x+\cdots+w_{m-1}x^{m-1-j},
\]
whose displayed exponents all lie between $0$ and $m-1$.  Hence
$\eps(x^{-j}w)=w_j\ne0$.  Thus no nonzero $w\in\R$ is orthogonal to all
of $\R$.

For general $n$, if $u=(u_1,\dots,u_n)$ is orthogonal to every vector,
fix $i$ and choose $v$ with an arbitrary $h\in\R$ in the $i$-th
coordinate and zeros elsewhere.  Then $\eps(u_i h)=0$ for every
$h\in\R$, so the case $n=1$ gives $u_i=0$.  This holds for every $i$.
\end{proof}

\begin{theorem}\label{thm:basic}
Let $C\subseteq\R^n$ be a QP code. Then:
\begin{enumerate}[label=\normalfont(\alph*),leftmargin=2.2em]
\item $C^{\circ}$ is a QP code;
\item $C^{\circ}=C^{\perp_{\mathcal R}}$;
\item if $C$ is generated by a set $S\subset \mathcal R^n$ as an $\R$-module, then $$C^\circ=\{v\in\R^n:u\cdot v=0\ \text{for all}\ u\in S\};$$
\item $|C|\,|C^{\circ}|=q^{mn}$;
\item $(C^{\circ})^{\circ}=C$.
\end{enumerate}
\end{theorem}

\begin{proof}
\begin{enumerate}
\item[(a)] Since the annihilator form is $\F_q$-bilinear,
$C^\circ$ is an $\F_q$-subspace of $\R^n$. Indeed, if
$v,w\in C^\circ$ and $a\in\F_q$, then for every $u\in C$,
$\ip{u}{v+w}=\ip{u}{v}+\ip{u}{w}=0$ and
$\ip{u}{av}=a\ip{u}{v}=0$. Thus $v+w\in C^\circ$ and
$av\in C^\circ$.

It remains to show that $C^\circ$ is closed under multiplication by
elements of $\R$. Let $v\in C^\circ$ and $r\in\R$. For any $u\in C$,
since $C$ is an $\R$-submodule, we have $ru\in C$. Hence, by
Lemma~\ref{lem:selfadj},
$\ip{u}{rv}=\ip{ru}{v}=0$,
because $ru\in C$ and $v\in C^\circ$. Therefore $rv\in C^\circ$.

Hence $C^\circ$ is an $\R$-submodule of $\R^n$, and so it is a
QP code.

\item[(b)] If $v\in C^{\perp_{\mathcal R}}$, then
$u\cdot v=0$ for every $u\in C$. Hence
$\ip{u}{v}=\eps(u\cdot v)=0$, and therefore
$v\in C^\circ$. Thus
$C^{\perp_{\mathcal R}}\subseteq C^\circ$.

For the reverse inclusion, let $v\in C^\circ$. Fix $u\in C$ and set
$w=u\cdot v\in\R$. Since $C$ is an $\R$-submodule, $hu\in C$ for every
$h\in\R$. As $v\in C^\circ$, we have
$\ip{hu}{v}=0$ for every $h\in\R$. Therefore
\[
0=\ip{hu}{v}
=\eps\bigl((hu)\cdot v\bigr)
=\eps\bigl(h(u\cdot v)\bigr)
=\eps(hw)
\]
for every $h\in\R$. By Proposition~\ref{prop:nondeg}, applied with
$n=1$, this implies $w=0$. Hence $u\cdot v=0$.

Since $u\in C$ was arbitrary, $v\in C^{\perp_{\mathcal R}}$. Thus
$C^\circ\subseteq C^{\perp_{\mathcal R}}$, and hence
$C^\circ=C^{\perp_{\mathcal R}}$.

\item[(c)] Suppose that $v\in\R^n$ satisfies
$u\cdot v=0$ for every $u\in S$. Since $S$ generates $C$ as an
$\R$-module, every element $u\in C$ is a finite $\R$-linear
combination of elements of $S$. Thus, for some
$u_1,\ldots,u_t\in S$ and $r_1,\ldots,r_t\in\R$, we have $u=r_1u_1+\cdots+r_tu_t.$
Therefore,
\[
u\cdot v
=
(r_1u_1+\cdots+r_tu_t)\cdot v
=
r_1(u_1\cdot v)+\cdots+r_t(u_t\cdot v)
=
0.
\]
Hence $u\cdot v=0$ for every $u\in C$, so
$v\in C^{\perp_{\mathcal R}}$. By part (b),
$v\in C^\circ$.

Conversely, let $v\in C^\circ$. By part (b),
$v\in C^{\perp_{\mathcal R}}$, and therefore
$u\cdot v=0$ for every $u\in C$. Since $S\subseteq C$, this holds in
particular for every $u\in S$. Hence
\[
C^\circ
=
\{v\in\R^n:u\cdot v=0\ \text{for all }u\in S\}.
\]

\item[(d)] Let $k=\dim_{\F_q}C$, and $\{c_1,\ldots,c_k\}$ be a basis of $C$.
Take $v\in\R^n$. Since $c_1,\ldots,c_k$ form a basis of $C$, we have
$v\in C^\circ$ if and only if
$\ip{c_1}{v}=0,\ldots,\ip{c_k}{v}=0$.

We show that these $k$ conditions are independent. Suppose that
$a_1,\ldots,a_k\in\F_q$ satisfy
$a_1\ip{c_1}{v}+\cdots+a_k\ip{c_k}{v}=0$
for every $v\in\R^n$. By bilinearity,
$\ip{a_1c_1+\cdots+a_kc_k}{v}=0$
for every $v\in\R^n$. Since the annihilator form is non-degenerate, the only vector
orthogonal to every $v\in\R^n$ is the zero vector. Hence
$a_1c_1+\cdots+a_kc_k=0$. Since $c_1,\ldots,c_k$ are linearly
independent, we get $a_1=\cdots=a_k=0$. Thus the $k$ conditions are
independent.

Since $\R^n$ has dimension $mn$ over $\F_q$, it follows that
$\dim_{\F_q}C^\circ=mn-k$. Therefore
$|C|=q^k$ and $|C^\circ|=q^{mn-k}$, so
$|C|\,|C^\circ|=q^{mn}$.

\item[(e)] Let $u\in C$. To show that $u\in(C^\circ)^\circ$, by definition we
need to show that $\ip{v}{u}=0$ for every $v\in C^\circ$.
Now, since $v\in C^\circ$ and $u\in C$, we have
$\ip{u}{v}=0$. As the annihilator form is symmetric,
$\ip{v}{u}=\ip{u}{v}=0$. Therefore
$u\in(C^\circ)^\circ$. Since this holds for every $u\in C$, we obtain
$C\subseteq(C^\circ)^\circ$.

By part (d),
$|C|\,|C^\circ|=q^{mn}$ and
$|C^\circ|\,|(C^\circ)^\circ|=q^{mn}$.
Therefore $|C|=|(C^\circ)^\circ|$.

Since $C\subseteq(C^\circ)^\circ$ and both have the same size,
$(C^\circ)^\circ=C$. \qedhere
\end{enumerate}
\end{proof}

\begin{corollary}
Let $C\subseteq\R^n$ be an $\F_q$-subspace. Then  $C$ is a QP code if and only if $C^{\circ}$ is a QP
code.
\end{corollary}

\begin{proof}
If $C$ is QP then so is $C^\circ$ by Theorem~\ref{thm:basic}(a).
Conversely, if $C^\circ$ is QP, then applying the same result to $C^\circ$
shows $(C^\circ)^\circ$ is QP, and $(C^\circ)^\circ=C$ by
Theorem~\ref{thm:basic}.
\end{proof}

\begin{remark}{\em
The restriction to subspaces cannot be dropped. Orthogonality against a
set is orthogonality against its $\F_q$-span, so for $q=2$,
$f(x)=x^{2}+x+1$, $n=2$ and $C=\{(1,0),(x,0)\}$ one has
$C^{\circ}=(\R\times\{0\})^{\circ}=\{0\}\times\R$, a QP code, although
$C$ is not a subspace.}
\end{remark}


\begin{remark}\label{rem:wss}{\em
Proposition~4.1 of \cite{WSS} states that the annihilator dual of a QP
code is again a QP code.  The proof appears to use
\[
   \ipf{u}{\lambda v}=\lambda(0)\ipf{u}{v},
\]
which is not valid for a general $\lambda\in\R$.  For example, over
$\F_2$ with $f(x)=x^2+x+1$, taking $u=1$, $v=x$ and $\lambda=x$ gives
$\ipf{u}{\lambda v}=\eps(x^2)=\eps(x+1)=1$, whereas the right-hand side
is $0$.  The statement itself remains true: if $v\in C^\circ$ and
$\lambda\in\R$, then for every $u\in C$,
$
   \ip{u}{\lambda v}=\ip{\lambda u}{v}=0
$
by Lemma~\ref{lem:selfadj}, since $\lambda u\in C$.  This is 
the argument used in Theorem~\ref{thm:basic}(a).  The same
non-degeneracy argument also gives $C^\circ=C^{\perp_{\mathcal R}}$ for every index.} 
\end{remark}

\section{The Gram matrix and the twisted dual}\label{sec:gram}


For $\ell\ge0$, we define
\[
s_\ell=\eps\bigl(x^\ell\bmod f\bigr).
\]
Then
$s_0=1$ and $s_1=\dots=s_{m-1}=0$. For $\ell\ge m$, using the relation \eqref{eq:xm} gives the linear recurrence
\begin{equation}
s_\ell=\sum_{i=0}^{m-1}f_i\,s_{\ell-m+i}\qquad(\ell\ge m).
\end{equation}
In particular, taking $\ell=m$ and using $s_0=1$, $s_1=\dots=s_{m-1}=0$, we get
\begin{equation}
s_m=f_0 .
\end{equation}
With respect to the basis $\mathcal B=\{1,x,\dots,x^{m-1}\}$, let $G=G_f=(G_{ij})_{0\le i,j\le m-1}$ be the matrix defined by $G_{ij}=s_{i+j}$.
Since the entry $G_{ij}$ depends only on $i+j$, the matrix $G$ is a \emph{Hankel}
matrix, and in particular symmetric.

\begin{lemma}\label{lem:gramR}
For all $g,h\in\R$,
\[\ipf{g}{h}=\psi(g)\,G\,\psi(h)^{\mathsf T}.\]
Consequently, $G$ is the Gram matrix of $\ipf{\cdot}{\cdot}$ relative to
the basis $\mathcal B$.
\end{lemma}

\begin{proof}
Write $\psi(g)=(g_0,\ldots,g_{m-1})$ and $\psi(h)=(h_0,\ldots,h_{m-1})$,
so that $g=\sum_{i}g_ix^{i}$ and $h=\sum_{j}h_jx^{j}$. By
$\F_q$-bilinearity of the form (Lemma~\ref{lem:bilinear} with $n=1$),
\[
\ipf{g}{h}
=\Bigl\langle\sum_{i=0}^{m-1}g_ix^i,\ \sum_{j=0}^{m-1}h_jx^j\Bigr\rangle_f
=\sum_{i=0}^{m-1}\sum_{j=0}^{m-1}g_ih_j\,\ipf{x^i}{x^j},
\]
and $\ipf{x^i}{x^j}=\eps(x^{i}x^{j})=\eps(x^{i+j})=s_{i+j}=G_{ij}$.
Hence
\[
\ipf{g}{h}=\sum_{i=0}^{m-1}\sum_{j=0}^{m-1}g_i\,G_{ij}\,h_j
=\psi(g)\,G\,\psi(h)^{\mathsf T}.\qedhere
\]
\end{proof}

\begin{proposition}\label{prop:blockgram}
Let $A=I_n\otimes G$, i.e.\ the block-diagonal $mn\times mn$ matrix with
$n$ diagonal blocks equal to $G$. Then for all $u,v\in\R^n$,
\[
\ip{u}{v}=\Psi(u)\,A\,\Psi(v)^{\mathsf T} .
\]
Thus $A$ is the Gram matrix of the annihilator form relative to the basis
of $\R^n$ underlying $\Psi$. Moreover $A$ is symmetric (respectively invertible) if and only if
 $G$ is symmetric (respectively invertible).
\end{proposition}

\begin{proof}
Write $u=(u_1,\dots,u_n)$ and $v=(v_1,\dots,v_n)$ with
$u_k=\sum_{i=0}^{m-1}a_{ki}x^i$ and $v_k=\sum_{j=0}^{m-1}b_{kj}x^j$, and
let $\alpha_k=(a_{k0},\dots,a_{k,m-1})$ and
$\beta_k=(b_{k0},\dots,b_{k,m-1})$ be the corresponding coordinate blocks,
so that $\Psi(u)=(\alpha_1,\dots,\alpha_n)$ and $\Psi(v)=(\beta_1,\dots,\beta_n).$

By Definition~\ref{def:form} and Lemma~\ref{lem:gramR},
\[
\ip{u}{v}=\sum_{k=1}^{n}\ipf{u_k}{v_k}
=\sum_{k=1}^{n}\alpha_k\,G\,\beta_k^{\mathsf T}.
\]
On the other hand, multiplying the row vector $\Psi(u)$ by the
block-diagonal matrix $A$ acts blockwise,
$\Psi(u)A=(\alpha_1G,\dots,\alpha_nG)$, and therefore
\[
\Psi(u)\,A\,\Psi(v)^{\mathsf T}
=(\alpha_1G,\dots,\alpha_nG)\,(\beta_1,\dots,\beta_n)^{\mathsf T}
=\sum_{k=1}^{n}\alpha_k G\beta_k^{\mathsf T}.
\]
Symmetry of $A$ follows from
symmetry of $G$, and $\det A=(\det G)^n$ gives the last statement.
\end{proof}



\begin{theorem}\label{thm:det}
The Hankel matrix $G$ is symmetric and invertible with determinant
\[
\det G=(-1)^{(m-1)(m-2)/2}\,f_0^{\,m-1}.
\]
\end{theorem}

\begin{proof}
Since $G_{ij}=s_{i+j}$, we have $G_{ij}=G_{ji}$, and hence $G$ is symmetric.
To compute the determinant, let $J$ be the anti-identity matrix and set
$H=GJ$. Right multiplication by $J$ reverses the columns of $G$. Thus,
the $j$-th column of $H$ is the $(m-1-j)$-th column of $G$, and therefore
\[
H_{ij}=G_{i,m-1-j}=s_{m-1+i-j}.
\]
Also, reversing $m$ columns requires $m(m-1)/2$ transpositions, so
$\det J=(-1)^{m(m-1)/2}$.
\vskip 2pt
Recall that $s_0=1$, $s_1=\cdots=s_{m-1}=0$, and $s_m=f_0$. For $i=0$,
we have $H_{0j}=s_{m-1-j}$, so the first row of $H$ is
$(0,\ldots,0,1)$. Now let $i\geq 1$. If $j\geq i$, then
$m-1+i-j\leq m-1$, and hence $H_{ij}=0$. On the other hand,
$H_{i,i-1}=s_m=f_0$. Thus $H$ has the form
\[
H=
\begin{pmatrix}
0   & 0   & \cdots & 0   & 1\\
f_0 & 0   & \cdots & 0   & 0\\
*   & f_0 & \cdots & 0   & 0\\
\vdots & \vdots & \ddots & \vdots & \vdots\\
*   & *   & \cdots & f_0 & 0
\end{pmatrix}.
\]

Expanding $\det H$ along the first row, the only nonzero entry is the
last entry, which is $1$. Hence
\[
\det H
= (-1)^{1+m}
\det
\begin{pmatrix}
f_0 & 0 & \cdots & 0\\
*   & f_0 & \cdots & 0\\
\vdots & \vdots & \ddots & \vdots\\
*   & * & \cdots & f_0
\end{pmatrix}
      =(-1)^{m-1}f_0^{m-1}.
\]

Since $H=GJ$, we have $\det H=\det G\,\det J$. Therefore
\[
\det G
=\frac{\det H}{\det J}
=(-1)^{(m-1)(m-2)/2}f_0^{\,m-1}.
\]
Finally, $f_0\neq 0$, so $\det G\neq 0$. Hence $G$ is invertible.
\end{proof}

\begin{lemma}\label{lem:twisteddual}
Let $V\subseteq\mathbb{F}_q^{N}$ be a subspace, with
$V^{\perp_E}$ denoting its Euclidean dual. Suppose that the matrix
$A\in\mathrm{GL}_N(\mathbb{F}_q)$ is symmetric. Then
$
(VA)^{\perp_E}=V^{\perp_E}A^{-1}.
$ 
\end{lemma}

\begin{proof}
Let $x\in\F_q^N$ be a row vector. By definition of the Euclidean dual,
\[
x\in(VA)^{\dpe}
\iff (vA)\,x^{\mathsf T}=0\ \text{ for all }v\in V .
\]
Rewrite the scalar $(vA)x^{\mathsf T}=v\,(Ax^{\mathsf T})
=v\,(xA^{\mathsf T})^{\mathsf T}=v\,(xA)^{\mathsf T}$, using
$A^{\mathsf T}=A$. Hence the condition reads
$v(xA)^{\mathsf T}=0$ for all $v\in V$, i.e.\ $xA\in V^{\dpe}$. Since $A$
is invertible this is equivalent to $x\in V^{\dpe}A^{-1}$.
\end{proof}

\begin{corollary}\label{cor:twist}
For every QP code $C\subseteq\R^n$, with $A=I_n\otimes G$,
\[
\Psi(C^{\circ})=\bigl(\Psi(C)\,A\bigr)^{\dpe}
=\Psi(C)^{\dpe}\,A^{-1}
=\Psi(C)^{\dpe}\,(I_n\otimes G^{-1}).
\]
Thus the annihilator dual is the ordinary Euclidean dual twisted by the
explicit, symmetric, block-diagonal matrix $I_n\otimes G^{-1}$.
\end{corollary}

\begin{proof}
By Proposition~\ref{prop:blockgram}, for $u,v\in\R^n$ we have
$\ip{u}{v}=\Psi(u)A\Psi(v)^{\mathsf T}=\bigl(\Psi(u)A\bigr)\Psi(v)^{\mathsf
T}$. Hence
\[
v\in C^{\circ}
\iff \bigl(\Psi(u)A\bigr)\Psi(v)^{\mathsf T}=0\ \ \forall u\in C
\iff \Psi(v)\in\bigl(\Psi(C)A\bigr)^{\dpe},
\]
which is the first equality. The second follows from
Lemma~\ref{lem:twisteddual} applied to $V=\Psi(C)$, since  $A$
is symmetric and invertible by Proposition~\ref{prop:blockgram} and
Theorem~\ref{thm:det}. The third is $A^{-1}=I_n\otimes G^{-1}$.
\end{proof}

\begin{definition}
    Let $\mathcal C$ and $\mathcal D$ be two $\mathbb F_q$-linear codes of length $N$, and let $A\in\mathrm{GL}_N(\mathbb{F}_q)$ be symmetric. If $\mathcal DA=\mathcal C^{\perp_E}$, then  we say that $\mathcal D$ is the $A$-\emph{twisted dual} of $\mathcal C$.
\end{definition}

The identity in Corollary~\ref{cor:twist} is the QP code version of  twisted-dual
description.  In contrast with monomial special cases, the block $G$ is
generally a non-monomial Hankel matrix; Section~\ref{sec:gray} studies
when a single change of basis on each block removes this twist.

\section{Coordinatewise duality-preserving Gray maps}\label{sec:gray}

Recall that $\psi:\R\to\F_q^m$ denotes the coefficient expansion with respect to the basis
$\mathcal B$.  For a matrix $S\in \mathrm{GL}_m(\mathbb{F}_q)$ and an element $a\in\R$, we define
$
   \phi_S(a)=\psi(a)S$,
and extend $\phi_S$ coordinatewise to
\[
   \Phi_S:\R^n\longrightarrow\F_q^{mn},\qquad
   \Phi_S(u)=\Psi(u)(I_n\otimes S).
\]
We use these maps as  the \emph{Gray maps}  in this section.  The
restriction to a common block matrix $S$ is essential: if arbitrary
matrices in $\mathrm{GL}_{mn}(\F_q)$ are allowed to mix different $\R$-coordinates,
the existence problem depends on the total dimension $mn$ rather than
only on the one-coordinate Gram matrix $G$.

\begin{definition}
The coordinatewise map $\Phi_S$ is \emph{duality-preserving} if
\[
   \Phi_S(C^\circ)=\Phi_S(C)^{\dpe}
\]
for every $\F_q$-subspace $C\subseteq\R^n$.
No Hamming-weight-preserving property is assumed.
\end{definition}

Duality-preserving basis maps over rings have been studied in a more general setting; see \cite{SU}. We also use the standard classification of non-degenerate symmetric bilinear forms over finite fields, recalled in the following lemma; see \cite{SL}.

\begin{lemma}\label{lem:congruence}
Let  $H\in \mathrm{GL}_m(\mathbb{F}_q)$ be symmetric and non-alternating.
\begin{enumerate}[label=\normalfont(\alph*),leftmargin=2.2em]
\item If $q$ is even, then $H\sim I_m$.
\item If $q$ is odd, then $H\sim I_m$ if and only if $\det H$ is a
square in $\F_q^\ast$.
\item If $q$ is odd, then $H\sim\mu I_m$ for some
$\mu\in\F_q^\ast$ if and only if either $m$ is odd, or $m$ is even and
$\det H$ is a square.
\end{enumerate}
\end{lemma}

\begin{proof}
\textbf{(a)}
Suppose that $q$ is even. Since $H$ is non-alternating, at least one
diagonal entry of $H$ is nonzero. Also, $H$ is invertible, so
$\operatorname{rank}(H)=m$. By Theorem~1 of \cite{SL}, $H$ has a factor
$B$ with $m$ columns; that is, $H=BB^{\mathsf T}$ with
$B\in\mathrm{GL}_m(\F_q)$. Hence
$B^{-1}H(B^{-1})^{\mathsf T}=I_m$, and therefore $H\sim I_m$.

\vskip 2pt
\textbf{(b)}
Now suppose that $q$ is odd. Since $H$ is invertible, we have
$\operatorname{rank}(H)=m$. By Theorem~2 of \cite{SL}, $H$ has a
factorization $H=BB^{\mathsf T}$ with $B\in\mathrm{GL}_m(\F_q)$ if and only if
$\det H$ is a square in $\F_q^\ast$. If $H=BB^{\mathsf T}$ with $B$ invertible, then
$B^{-1}H(B^{-1})^{\mathsf T}=I_m$, so $H\sim I_m$. 

Conversely, if
$H\sim I_m$, then there exists $P\in\mathrm{GL}_m(\F_q)$ such that
$PHP^{\mathsf T}=I_m$. Hence $H=P^{-1}(P^{-1})^{\mathsf T},$
so $\det H=(\det P^{-1})^2$ is a square in $\F_q^\ast$.
Therefore, $H\sim I_m$ if and only if $\det H$ is a square in $\F_q^\ast$.

\vskip 2pt
\textbf{(c)} Observe that $H\sim\mu I_m$ if and only if
$\mu^{-1}H\sim I_m$. By part (b), this is equivalent to
$\mu^{-m}\det H$ being a square in $\F_q^\ast$.

If $m$ is odd, then $m-1$ is even. Hence
\[
\mu^m=\mu\cdot\mu^{m-1}
      =\mu\left(\mu^{(m-1)/2}\right)^2.
\]
Thus 
$\mu^m$ is a square if and only if $\mu$ is a square. Therefore we may
choose $\mu$ to be a square or a non-square according as $\det H$ is a
square or a non-square. Then $\mu^{-m}\det H$ is a square, so such a
$\mu$ always exists.

If $m$ is even, then $\mu^m=(\mu^{m/2})^2$ is always a square. Hence
$\mu^{-m}\det H$ is a square if and only if $\det H$ is a square.
\end{proof}
\begin{theorem}\label{thm:gray}
Let $G$ be the one-coordinate Gram matrix in
Lemma~\ref{lem:gramR}, and let $\Phi_S$ be the map above.
\begin{enumerate}[label=\normalfont(\alph*),leftmargin=2.2em]
\item $\Phi_S$ is a duality-preserving map if and only if
$  SS^{\mathsf T}=\lambda G$ for some $\lambda\in\F_q^\ast$.
\item If $q$ is even, a coordinatewise duality-preserving map always
exists.
\item If $q$ is odd and $m$ is odd, such a map always exists.
\item If $q$ is odd and $m$ is even, such a map exists if and only if
\[
   \det G=(-1)^{(m-1)(m-2)/2}f_0^{\,m-1}
\]
is a square in $\F_q^\ast$; equivalently,
$(-1)^{(m-2)/2}f_0$ is a square in $\F_q^\ast$.
\end{enumerate}
\end{theorem}

\begin{proof}
We set $A=I_n\otimes G$ and $B=I_n\otimes(SS^{\mathsf T})$.
By Proposition~\ref{prop:blockgram}, we have \[\ip{u}{v}=\Psi(u)A\Psi(v)^{\mathsf T}.\]

Note that $\Phi_S(u)=\Psi(u)(I_n\otimes S)$ and $\Phi_S(v)=\Psi(v)(I_n\otimes S).$
Therefore
\[
\begin{aligned}
\Phi_S(u)\Phi_S(v)^{\mathsf T}
&=\Psi(u)(I_n\otimes S)
   \bigl(\Psi(v)(I_n\otimes S)\bigr)^{\mathsf T}\\
&=\Psi(u)(I_n\otimes S)(I_n\otimes S^{\mathsf T})\Psi(v)^{\mathsf T}\\
&=\Psi(u)(I_n\otimes SS^{\mathsf T})\Psi(v)^{\mathsf T}\\
&=\Psi(u)B\Psi(v)^{\mathsf T}.
\end{aligned}
\]

\textbf{(a)} If $SS^{\mathsf T}=\lambda G$, then $B=\lambda A$. Hence, for any
$u,v$,
\[
\Phi_S(u)\Phi_S(v)^{\mathsf T}=\Psi(u)B\Psi(v)^{\mathsf T}=\lambda\Psi(u)A\Psi(v)^{\mathsf T}=\lambda\ip{u}{v}.
\]

Since $\lambda\neq0$, $\ip{u}{v}=0$ if and only if  $\Phi_S(u)\Phi_S(v)^{\mathsf T}=0.$
Thus, $u$ and $v$ are orthogonal if and only if $\Phi_S(u)$ and
$\Phi_S(v)$ are orthogonal.

\vskip 2pt

Conversely, suppose that $\Phi_S$ is duality-preserving. Let
$x\in\F_q^{mn}$ be nonzero, and choose $u$ such that $\Psi(u)=x$.
Consider the one-dimensional subspace $\langle u\rangle$.

Let $z=\Psi(v)$. Then
\[
z\in\ker(xA)
\iff xAz^{\mathsf T}=0 \iff \Psi(u)A\Psi(v)^{\mathsf T}=0 \iff \ip{u}{v}=0.
\]
Since $\Phi_S$ is duality-preserving,
$\ip{u}{v}=0$ if and only if
$\Phi_S(u)\Phi_S(v)^{\mathsf T}=0$. Using the definition of $B$, this gives
\[
\ip{u}{v}=0
\iff \Phi_S(u)\Phi_S(v)^{\mathsf T}=0\iff \Psi(u)B\Psi(v)^{\mathsf T}=0\iff xBz^{\mathsf T}=0\iff z\in\ker(xB).
\]
Therefore $\ker(xA)=\ker(xB)$.

Since $A$ and $B$ are invertible and $x\neq0$, both $xA$ and $xB$ are
nonzero linear functionals. Two nonzero linear functionals with the same
kernel differ by a nonzero scalar. Hence,  $xB=\lambda_x xA$ for some
$\lambda_x\in\F_q^\ast$.

It remains to show that the scalar $\lambda_x$ is independent of $x$.
Let $x$ and $y$ be nonzero and linearly independent. Then $xA$ and
$yA$ are also linearly independent, since $A$ is invertible. We have $xB=\lambda_x xA$ and  $yB=\lambda_y yA$
and, applying the same relation to $x+y$,
\[
(x+y)B=\lambda_{x+y}(x+y)A.
\]
But $(x+y)B=xB+yB
      =\lambda_x xA+\lambda_y yA$.
Comparing the two expressions, we get 
\begin{align*}
    \lambda_{x+y}(x+y)A=\lambda_{x+y}xA+\lambda_{x+y}yA\\
(\lambda_x-\lambda_{x+y})xA
+(\lambda_y-\lambda_{x+y})yA=0.
\end{align*}

Since $xA$ and $yA$ are linearly independent, we get
$\lambda_x=\lambda_y=\lambda_{x+y}$.
\vskip 2pt \noindent
If $x$ and $y$ are linearly dependent, say $y=cx$ with
$c\in\F_q^\ast$, then $yB=cxB=c\lambda_x xA=\lambda_x yA.$
On the other hand, $yB=\lambda_y yA$. Comparing the two expressions gives 
\[\lambda_x yA=\lambda_y yA \implies (\lambda_x-\lambda_y) yA=0.\]
Since $yA\neq0$, we obtain
$\lambda_y=\lambda_x$. Thus the same scalar $\lambda\in\F_q^\ast$ works for every nonzero
$x\in\F_q^{mn}$. Hence $xB=\lambda xA$ for every $x$, and therefore
$B=\lambda A$. 

\vskip 3pt

Since
$A=I_n\otimes G$ and
$B=I_n\otimes(SS^{\mathsf T})$, comparison of the diagonal blocks gives
$SS^{\mathsf T}=\lambda G.$
This proves (a).
\vskip 2pt
\textbf{(b)} Suppose that $q$ is even. 
Since $G$ is invertible, it is non-degenerate. Also,
$G_{00}=s_0=1\neq0$. Hence $G$ is non-alternating. Hence, by Lemma~\ref{lem:congruence}(a),
$G\sim I_m$. Thus there exists $T\in \mathrm{GL}_m(\F_q)$ such that
\[
TGT^{\mathsf T}=I_m
\implies
G=T^{-1}(T^{-1})^{\mathsf T}.
\]
Taking $S=T^{-1}$ gives
$SS^{\mathsf T}=G$. Thus the condition in part (a) holds with
$\lambda=1$, and hence a coordinatewise duality-preserving map exists.
\vskip 2pt
\textbf{(c)} Suppose that $q$ is odd and $m$ is odd. 
As above, $G$ is symmetric, non-degenerate, and non-alternating.
Since $q$ is odd and $m$ is odd, Lemma~\ref{lem:congruence}(c) gives
$\lambda\in\F_q^\ast$ such that
$\lambda G\sim I_m$. Hence there exists $T\in \mathrm{GL}_m(\F_q)$ such that
\[
T(\lambda G)T^{\mathsf T}=I_m\implies
\lambda G=T^{-1}(T^{-1})^{\mathsf T}.
\]
Taking $S=T^{-1}$ gives
$SS^{\mathsf T}=\lambda G$. Thus the condition in part (a) holds, and
a coordinatewise duality-preserving map exists.
\vskip 2pt
\textbf{(d)} Suppose that $q$ is odd and $m$ is even.
By part (a), a coordinatewise duality-preserving map exists if and
only if there exist $S\in \mathrm{GL}_m(\F_q)$ and
$\lambda\in\F_q^\ast$ such that $SS^{\mathsf T}=\lambda G.$
Since $S$ is invertible, $SS^{\mathsf T}=\lambda G
\iff
\lambda G\sim I_m.$
Since $q$ is odd, Lemma~\ref{lem:congruence}(b) gives
\[
\lambda G\sim I_m
\iff
\det(\lambda G)\text{ is a square in }\F_q^\ast.
\]

We know that $\det(\lambda G)=\lambda^m\det(G).$ Since $m$ is even,
$\lambda^m=(\lambda^{m/2})^2$ is a square in $\F_q^\ast$. Hence
\[
\det(\lambda G)\text{ is a square}
\iff
\det(G)\text{ is a square}.
\]
Thus, a coordinatewise duality-preserving map exists if and only if
$\det(G)$ is a square in $\F_q^\ast$.
By Theorem~\ref{thm:det}, we have 
\[
\det(G)
=
(-1)^{(m-1)(m-2)/2}f_0^{m-1}.
\]

Since $m$ is even, $m-1$ is odd, we note that $f_0^{m-1}
=
f_0\left(f_0^{(m-2)/2}\right)^2$.
Hence $f_0^{m-1}$ is a square if and only if $f_0$ is a square.
Also, $(-1)^{(m-1)(m-2)/2}
=
(-1)^{(m-2)/2},$
because $m-1$ is odd.

Thus, ignoring the square factor
$\left(f_0^{(m-2)/2}\right)^2$, we have
\[
\det(G)
\equiv
(-1)^{(m-2)/2}f_0
\pmod{(\mathbb F_q^\ast)^2}.
\]

Therefore, $\det(G)$ is a square if and only if $(-1)^{(m-2)/2}f_0$ is a square. Thus, when $q$ is odd and $m$ is even, a coordinatewise
duality-preserving map exists if and only if $(-1)^{(m-2)/2}f_0$
is a square in $\F_q^\ast.$
\end{proof}

\begin{example}{\em
Let $q=3$, $m=2$ and
$f(x)=x^2-(f_1x+2)$. Here $f_0=2$, and in $\R$ we have
$x^2=f_1x+2$. Hence
$s_0=1$, $s_1=0$ and $s_2=\eps(x^2)=2$. Therefore
\[
G=\begin{pmatrix}
1&0\\
0&2
\end{pmatrix},
\qquad
\det G=2.
\]
The only nonzero square in $\F_3$ is $1$, so $2$ is not a square in
$\F_3^\ast$. Since $q$ is odd and $m=2$ is even,
Theorem~\ref{thm:gray} shows that no coordinatewise
duality-preserving Gray map exists.}
\end{example}

\begin{remark}{\em
The word ``duality-preserving'' refers only to orthogonal complements.
An invertible block matrix $S$ need not preserve Hamming weight, so the
Hamming parameters of $\Phi_S(C)$ may depend on the chosen admissible
$S$.  This distinction is important in the CSS examples below.}
\end{remark}

\section{Structure of annihilator duals and self-orthogonality}
\label{sec:duals}

So far the annihilator dual has been described abstractly and through its
Gram matrix. In this section, we give two descriptions of $C^\circ$. 
First, we express it as a syzygy module over $\R$. Second, when $f$ is squarefree, we show that annihilator duality decouples completely across the Chinese-Remainder factors and becomes \emph{Euclidean} duality
on each factor. Together, these reduce questions
about self-orthogonal, self-dual, LCD and dual-containing QP codes to
classical questions over the residue fields, and in the one-generator case
to a single polynomial congruence.

\begin{theorem}
Let $C\subseteq\R^n$ be generated as an $\R$-module by the rows of a
matrix $M\in\R^{\rho\times n}$. Then
\[
C^\circ=\{v\in\R^n:Mv^{\mathsf T}=0\text{ in }\R^\rho\},
\]
the syzygy module of the rows of $M$.
\end{theorem}

\begin{proof}
By Theorem~\ref{thm:basic}(b) we have $C^{\circ}=C^{\perp_{\mathcal R}}$, so
$v\in C^{\circ}$ if and only if $u\cdot v=0$ for every $u\in C$. By  Theorem~\ref{thm:basic}(c), it suffices to impose this for
the generators, i.e.\ for the rows $u_1,\dots,u_\rho$ of $M$. The
condition $u_t\cdot v=0$ for $t=1,\dots,\rho$ is precisely the statement
that the $t$-th entry of the column vector $Mv^{\mathsf T}$ vanishes.
\end{proof}


\begin{theorem}\label{thm:global}
Assume that $f=\prod_{j=1}^{k}p_j$ is squarefree, let $\R\cong\bigoplus_j\K_j$ be the decomposition \eqref{eq:crt} into fields $\K_j=\F_q[x]/\langle p_j\rangle$, and let $C\cong\bigoplus_j C_j$ be the corresponding decomposition  \eqref{eq:Cdecomp} of a QP code, with
$C_j\subseteq\K_j^{\,n}$. Then
\[
C^{\circ}\;\cong\;\bigoplus_{j=1}^{k}(C_j)^{\dpe},
\]
where $(C_j)^{\dpe}$ denotes the ordinary Euclidean dual of $C_j$ taken
over the field $\K_j$.
\end{theorem}

\begin{proof}
By Theorem~\ref{thm:basic}(b), $v\in C^{\circ}$ if and only if $u\cdot
v=0$ in $\R$ for all $u\in C$. Write $u\leftrightarrow(u^{(1)},\dots,
u^{(k)})$ and $v\leftrightarrow(v^{(1)},\dots,v^{(k)})$ under
\eqref{eq:crt}, where $u^{(j)},v^{(j)}\in\K_j^{\,n}$. Under the ring isomorphism \eqref{eq:crt}, multiplication and addition are computed componentwise, so the $j$-th component of $u\cdot v$ is 
\[
(u\cdot v)^{(j)}=\sum_{i=1}^{n}u^{(j)}_iv^{(j)}_i
=u^{(j)}\cdot v^{(j)},
\]
the standard $\K_j$-bilinear dot product on $\K_j^{\,n}$. An element of
$\R$ vanishes if and only if all its components vanish, so
\[
u\cdot v=0
\iff
u^{(j)}\cdot v^{(j)}=0\ \text{ for every }j .
\]
Since $u$ ranges over $C$ exactly when $u^{(j)}$ ranges over $C_j$ for
each $j$ independently,  $v\in C^\circ$ if and
only if $v^{(j)}\in(C_j)^{\dpe}$ for every $j$.
\end{proof}

In the case $n=1$ this is compatible with the product-ring decomposition in
\cite{AS}; the argument above shows that the same componentwise
description holds for arbitrary index.  Standard self-orthogonal,
self-dual, LCD and dual-containing conditions may therefore be checked
constituent by constituent.

\begin{definition} A code  $C$ is \emph{annihilator
self-orthogonal} if $C\subseteq C^{\circ}$, \emph{annihilator self-dual}
if $C=C^{\circ}$, \emph{annihilator LCD} if $C\cap C^{\circ}=\{0\}$, and
\emph{annihilator dual-containing} if $C^{\circ}\subseteq C$.\end{definition}

\begin{corollary}\label{thm:selforth}
Let $C\subseteq\R^n$ be a QP code.
\begin{enumerate}[label=\normalfont(\alph*),leftmargin=2.2em]
\item If $C$ is generated by the rows of $M\in\R^{\rho\times n}$, then $C$
is annihilator self-orthogonal if and only if
\[
MM^{\mathsf T}=0\quad\text{in }\R^{\rho\times\rho}.
\]
\item If $C=\langle a\rangle$ is generated by the single polynomial vector
$a=(a_1,\dots,a_n)\in\R^n$, then $C$ is annihilator self-orthogonal if and
only if
\[
\sum_{i=1}^{n}a_i(x)^{2}\equiv0\pmod{f(x)} .
\]
\item Assume the polynomial $f$ is squarefree with irreducible factors $p_j$ and let $C_j$ be the
constituents of $C$. Then $C$ is annihilator self-dual if and only if each
$C_j$ is Euclidean self-dual over $\K_j$; annihilator dual-containing if
and only if $(C_j)^{\dpe}\subseteq C_j$ for every $j$; and annihilator LCD
if and only if each $C_j$ is Euclidean LCD over $\K_j$.
\end{enumerate}
\end{corollary}

\begin{proof}
\begin{enumerate}
\item[(a)] By Theorem~\ref{thm:basic}(c), $C\subseteq C^{\circ}=C^{\perp_{\mathcal R}}$ means
that $u\cdot u'=0$ for all $u,u'\in C$. By $\R$-bilinearity of the dot
product it is enough to impose this on generators, i.e.\ on the rows
$u_1,\dots,u_\rho$ of $M$. The condition $u_t\cdot u_{t'}=0$ for all
$t,t'$ with $1\le t,t'\le\rho$ entails exactly that every entry of the matrix $MM^{\mathsf T}$
vanishes.
\item[(b)]  Apply (a) with $\rho=1$ and $M=a$, a single row. Then $MM^{\mathsf
T}$ is the $1\times1$ matrix with entry $a\cdot a=\sum_{i=1}^{n}a_i^{2}\in\R.$
This vanishes in $\R=\F_q[x]/\langle f\rangle$ if and only if $f$ divides
$\sum_i a_i(x)^{2}$ in $\F_q[x]$, which is the stated congruence.
\item[(c)] Immediate from Theorem~\ref{thm:global}: under the orthogonal
decomposition $C\cong\bigoplus_jC_j$ and
$C^{\circ}\cong\bigoplus_j(C_j)^{\dpe}$, an equality, inclusion or trivial
intersection holds for $C$ and $C^\circ$ if and only if it holds
componentwise.
\end{enumerate}
\end{proof}

Part~(b) gives the one-generator condition asked for in
\cite[\S6]{WSS}; parts (a) and (c) are direct consequences of
Theorem~\ref{thm:basic} and the CRT description.

\begin{remark}{\em
If $q$ is even, then
$\sum_i a_i^2=(\sum_i a_i)^2$.  Writing
$f=\prod_j p_j^{e_j}$, Corollary~\ref{thm:selforth}(b) is therefore
equivalent to
\[
   \prod_j p_j^{\lceil e_j/2\rceil}\ \Bigm|\ \sum_i a_i.
\]
This form is convenient for searching for one-generator
self-orthogonal QP codes.}
\end{remark}

\begin{proposition}\label{prop:dim}
Let $f$ be squarefree, and let $C_j\subseteq\K_j^n$ be the constituents of a given code $C$ as above. Then
\[
\dim_{\F_q}\Psi(C)=\sum_{j=1}^{k}\deg(p_j)\cdot\dim_{\K_j}C_j ,
\]
and the minimal number of $\R$-module generators  of $C$ equals $\max_j\dim_{\K_j}C_j$. In
particular the $\F_q$-dimension of $\Psi(C)$ is determined by the
constituent data, independent of the generator matrix.
\end{proposition}

\begin{proof}
The decomposition in Equation \eqref{eq:Cdecomp} is an isomorphism of
$\F_q$-vector spaces, so dimensions add:
$\dim_{\F_q}C=\sum_j\dim_{\F_q}C_j$. Each $C_j$ is a $\K_j$-space and
$\K_j=\F_{q^{\deg p_j}}$, so $\dim_{\F_q}C_j=\deg(p_j)\dim_{\K_j}C_j$.
Adding up provides the formula, and $\dim_{\F_q}\Psi(C)=\dim_{\F_q}C$ because
$\Psi$ is an $\F_q$-isomorphism. 
\vskip 2pt
For the generator-number statement, a minimal
generating set of $C\cong\bigoplus_jC_j$ over
$\R\cong\bigoplus_j\K_j$ needs $\max_j\dim_{\K_j}C_j$ elements, since one
may generate all constituents simultaneously by combining generators
through the idempotents of \eqref{eq:crt}.
\end{proof}

Proposition~\ref{prop:dim} above is included  only  to record the parameters in
the equal-block setting considered here.  Related constituent and
multi-generator formulas for generalized QP codes appear in
\cite{SPD}.

\section{The MacWilliams transform} \label{sec:mw}
In this section, we show how the annihilator dual interacts with the Hamming weight enumerator.
The twisted description in Corollary~\ref{cor:twist}  gives the
corresponding relation between Hamming weight enumerators. For an
$\F_q$-linear code $D\subseteq\F_q^{mn}$, write
\begin{align*}
W_D(X,Y)=\sum_{c\in D}X^{mn-\wt(c)}Y^{\wt(c)}.
\end{align*}
The following theorem is the classical MacWilliams identity applied to
the twisted coefficient image of a QP code.

\begin{theorem}\label{thm:mw}
For every QP code $C\subseteq\R^n$,
\begin{align*}
W_{\Psi(C^{\circ})}(X,Y)
=\frac{1}{|C|}\,
W_{\Psi(C)(I_n\otimes G)}\bigl(X+(q-1)Y,\,X-Y\bigr).
\end{align*}
If a coordinatewise duality-preserving Gray map $\Phi$ exists, then
\begin{align*}
W_{\Phi(C^{\circ})}(X,Y)
=\frac{1}{|C|}\,
W_{\Phi(C)}\bigl(X+(q-1)Y,\,X-Y\bigr).
\end{align*}
\end{theorem}

\begin{proof}
By Corollary~\ref{cor:twist},
$\Psi(C^{\circ})=\bigl(\Psi(C)(I_n\otimes G)\bigr)^{\dpe}$.
Here we use the first equality in Corollary~\ref{cor:twist}. The
equivalent expression
$\Psi(C^{\circ})=\Psi(C)^{\dpe}(I_n\otimes G^{-1})$
places the inverse twist after taking the Euclidean dual. Thus the code
to which we apply the MacWilliams identity is
$D=\Psi(C)(I_n\otimes G)$. Since $I_n\otimes G$ is invertible,
$D$ is an $\F_q$-linear code of length $mn$ and
$|D|=|\Psi(C)|=|C|$. The classical MacWilliams identity gives
\begin{align*}
W_{D^{\dpe}}(X,Y)
=\frac{1}{|D|}W_D\bigl(X+(q-1)Y,\,X-Y\bigr).
\end{align*}
Substituting $D^{\dpe}=\Psi(C^{\circ})$ and $|D|=|C|$ gives the first
formula.

Now suppose that $\Phi$ is duality-preserving. Then
$\Phi(C^{\circ})=\Phi(C)^{\dpe}$. Since $\Phi$ is an invertible
$\F_q$-linear map, $|\Phi(C)|=|C|$, and applying the same MacWilliams
identity to $\Phi(C)$ gives the second formula.
\end{proof}

Theorem~\ref{thm:mw}  above provides a MacWilliams identity for the Hamming
weight enumerators of the corresponding $\F_q$-images. In what follows, we maintain  more information. For each coordinate of a vector
in $\R^n$, we record its Gray weight, that is, the Hamming weight of
its image in $\F_q^m$. Thus, instead of recording only the total Gray
weight, we record how many coordinates have Gray weight
$0,1,\ldots,m$. This gives a partition-level version of the
MacWilliams identity. The general character-transform framework for
partition weight distributions over finite Frobenius rings is standard;
see Honold and Landjev \cite{HL}.  Here we write the specialization
explicitly because the two partitions are determined by the Gram twist
appearing in our annihilator pairing.

For a fixed $S\in\mathrm{GL}_m(\F_q)$,  we define $\phi_S(a)=\psi(a)S$ for 
$a\in\R$, as in Section~\ref{sec:gray}.  We also define a second map $\tau_S(a)=\psi(a)GS^{-\mathsf T}.$ Since $S$ and $G$ are invertible, both $\phi_S$ and $\tau_S$ are
$\F_q$-linear bijections from $\R$ onto $\F_q^m$.
The reason for introducing these two maps is the following simple
relation  that arises using the symmetry of $G$ and Lemma~\ref{lem:gramR}, for $a,b\in\R$:
\begin{equation} \label{eq:pairingsplit}
    \eps(ab)=\psi(a)G\psi(b)^{\mathsf T} =\bigl(\psi(a)S\bigr)
  \bigl(\psi(b)GS^{-\mathsf T}\bigr)^{\mathsf T} =\phi_S(a)\cdot\tau_S(b).
\end{equation}

Thus the annihilator pairing on $\R$ becomes the ordinary Euclidean
pairing on $\F_q^m$, with $\phi_S$ applied to the first argument and
$\tau_S$ to the second. This is a reformulation of
Lemma~\ref{lem:gramR} under the coordinate maps $\phi_S$ and $\tau_S$.
For elements in $\R^n$, the same relation is applied separately in each
of the $n$ coordinates, as in Proposition~\ref{prop:blockgram}.

We now use these two maps to define two weights on $\R$. For $a\in\R$, we denote $\wt_S(a)=\wt\bigl(\phi_S(a)\bigr)$ and 
 $\widetilde{\wt}_S(a)=\wt\bigl(\tau_S(a)\bigr),$
where $\wt$ denotes the usual Hamming weight in $\F_q^m$.
These weights give two partitions of the ring $\R$:
\[
\mathcal P_S
 =\bigl(\{a\in\R:\wt_S(a)=i\}\bigr)_{i=0}^{m},
\hbox{ and }
\mathcal Q_S
 =\bigl(\{a\in\R:\widetilde{\wt}_S(a)=j\}\bigr)_{j=0}^{m}.
\]
Since both $\phi_S$ and $\tau_S$ are onto $\F_q^m$, each of the
weights $0,1,\ldots,m$ occurs. Hence both partitions have $m+1$
nonempty parts.

There is a simple case in which these two partitions are the same.
Suppose that
$SS^{\mathsf T}=\lambda G$
for some $\lambda\in\F_q^{*}$.  This is
the same  condition as  $\Phi_S$ being a duality-preserving map (see Theorem~\ref{thm:gray}) and, in this case, $G=\lambda^{-1}SS^{\mathsf T},$
and hence
\[
\tau_S(a)
 =\psi(a)GS^{-\mathsf T}
 =\lambda^{-1}\psi(a)S
 =\lambda^{-1}\phi_S(a).
\]
Notice that multiplication by the nonzero scalar $\lambda^{-1}$ does not change
Hamming weight and, therefore
$\widetilde{\wt}_S(a)=\wt_S(a)$
for every $a\in\R$, and consequently
$\mathcal P_S=\mathcal Q_S.$

 More generally, we can write $\tau_S(a)=\phi_S(a)M,$ where
$M=S^{-1}GS^{-\mathsf T}.$
Since $\phi_S$ is onto $\F_q^m$, the two partitions coincide precisely
when $\wt(zM)=\wt(z)$ for every $z\in\F_q^m.$

An invertible linear map preserves the Hamming weight if and only if its
matrix is monomial. Thus $\mathcal P_S=\mathcal Q_S$
if and only if $M=S^{-1}GS^{-\mathsf T}$ is a monomial matrix.
We observe that this condition is slightly weaker than the existence of a
duality-preserving Gray map. 

We next define the corresponding composition enumerators. Let
$u=(u_1,\ldots,u_n)\in\R^n.$ For each $0\le i\le m$, let
$$N_i(u)
 =\bigl|\{k:\wt_S(u_k)=i\}\bigr|.$$
 $N_i(u)$ counts the number of coordinates of $u$ having
$\wt_S$-weight $i$, and
\[
N_0(u)+\cdots+N_m(u)=n.
\]
Similarly, for $v=(v_1,\ldots,v_n)$, define 
$
\widetilde N_j(v)
 =\bigl|\{k:\widetilde{\wt}_S(v_k)=j\}\bigr|
$, for $0\le i\le m$.

Given  a code $C\subseteq\R^n$, we define
\begin{equation}\label{eq:A}
A_C(X_0,\ldots,X_m)
 =\sum_{u\in C}
   \prod_{i=0}^{m}X_i^{N_i(u)} =\sum_{u\in C}\prod_{k=1}^{n}X_{\wt_S(u_k)}.
\end{equation}
In the same way, we define
\begin{equation}\label{eq:B}
B_C(Y_0,\ldots,Y_m)
 =\sum_{v\in C}
   \prod_{j=0}^{m}Y_j^{\widetilde N_j(v)}
 =\sum_{v\in C}\prod_{k=1}^{n}
   Y_{\widetilde{\wt}_S(v_k)}.
\end{equation}

For example, the coefficient  corresponding to the term 
$
X_0^{w_0}X_1^{w_1}\cdots X_m^{w_m}
$ 
in $A_C$ is the number of codewords of $C$ having exactly $w_i$
coordinates of $\wt_S$-weight $i$, for $0\le i\le m$. Thus $A_C$
records more information than the ordinary two-variable Gray-weight
enumerator. 
 The following result gives the MacWilliams relation between the two
composition enumerators. We  use the $q$-ary Krawtchouk numbers of length $m$ \cite[Ch.~5]{MS},
\[
K_j(i)
 =
 \sum_{h=0}^{j}
 (-1)^h(q-1)^{j-h}
 \binom{i}{h}\binom{m-i}{j-h},
\qquad 0\le i,j\le m,
\]
with the usual convention $\binom{a}{b}=0$ when $b<0$ or $b>a$. The following result provides McWilliams identities for the annhihilator dual. 

\begin{theorem}\label{thm:mwpart}
For every QP code $C\subseteq\R^n$ and every
$S\in\mathrm{GL}_m(\F_q)$,
\begin{align*}
B_{C^{\circ}}(Y_0,\ldots,Y_m)
=\frac{1}{|C|}
A_C\biggl(
&\sum_{j=0}^{m}K_j(0)Y_j,\,
\sum_{j=0}^{m}K_j(1)Y_j,\,
\ldots,\sum_{j=0}^{m}K_j(m)Y_j
\biggr).
\end{align*}
\end{theorem}

\begin{proof}
Fix a nontrivial additive character
$\chi:\F_q\longrightarrow\mathbb C^{*}.$
For a fixed $v\in\R^n$, consider the map $L_v:C\longrightarrow\F_q,$ defined by 
$L_v(u)=\eps(u\cdot v).$
This is an $\F_q$-linear map.
If $v\in C^{\circ}$, then  $\eps(u\cdot v)=0$
for every $u\in C.$
Hence
\[
\sum_{u\in C}\chi\bigl(\eps(u\cdot v)\bigr)
=\sum_{u\in C}\chi(0)
=|C|.
\]
Now suppose that $v\notin C^{\circ}$. Then $L_v$ is a nonzero
$\F_q$-linear map from $C$ to the one-dimensional space $\F_q$ and
hence $L_v(C)=\F_q$. For any $\alpha\in\F_q$, choose $u_\alpha\in C$ such that
$L_v(u_\alpha)=\alpha$. Then the solutions of $L_v(u)=\alpha$ are
exactly the vectors $u_\alpha+w$ for  $w\in\ker L_v$.
Thus every $\alpha\in\F_q$ has the same number of preimages, namely
$|\ker L_v|=|C|/q$. Therefore
\[
\sum_{u\in C}\chi\bigl(\eps(u\cdot v)\bigr)
=
\frac{|C|}{q}\sum_{\alpha\in\F_q}\chi(\alpha)
=0,
\]
since $\chi$ is a nontrivial additive character.
Therefore, we have shown that
\begin{equation}\label{eq:charindicator}
\sum_{u\in C}\chi\bigl(\eps(u\cdot v)\bigr)
=
\begin{cases}
|C|, & v\in C^{\circ},\\
0,   & v\notin C^{\circ}.
\end{cases}
\end{equation}
By Equations~(\ref{eq:B}),(\ref{eq:charindicator}) and  extending the sum for  all $v\in\R^n$, we get 
\begin{align*}
&\sum_{v\in\R^n}
\left(
\sum_{u\in C}\chi\bigl(\eps(u\cdot v)\bigr)
\right)
\prod_{k=1}^{n}
Y_{\widetilde{\wt}_S(v_k)}\\
&~~~~~~~~~~~=
\sum_{v\in C^\circ}
\left(
\sum_{u\in C}\chi\bigl(\eps(u\cdot v)\bigr)
\right)
\prod_{k=1}^{n}
Y_{\widetilde{\wt}_S(v_k)}
+
\sum_{v\notin C^\circ}
\left(
\sum_{u\in C}\chi\bigl(\eps(u\cdot v)\bigr)
\right)
\prod_{k=1}^{n}
Y_{\widetilde{\wt}_S(v_k)}\\\\
&~~~~~~~~~~~=
|C|B_{C^\circ}(Y_0,\ldots,Y_m) + 0.
\end{align*}
Interchanging the two finite sums, we obtain
\begin{align*}
|C|\,B_{C^{\circ}}(Y_0,\ldots,Y_m)
&=
\sum_{u\in C}
\sum_{v\in\R^n}
\chi\bigl(\eps(u\cdot v)\bigr)
\prod_{k=1}^{n}
Y_{\widetilde{\wt}_S(v_k)}.
\end{align*}
Hence, for fixed $u=(u_1,\ldots,u_n)\in C$, since $\eps$ is
$\F_q$-linear and $\chi$ is an additive character,
\[
\chi\bigl(\eps(u\cdot v)\bigr)
=
\chi\left(\sum_{k=1}^{n}\eps(u_kv_k)\right)
=
\prod_{k=1}^{n}\chi\bigl(\eps(u_kv_k)\bigr).
\]
Therefore
\begin{align*}
&\sum_{v\in\R^n}
\chi\bigl(\eps(u\cdot v)\bigr)
\prod_{k=1}^{n}Y_{\widetilde{\wt}_S(v_k)}
\\
&=
\sum_{v_1\in\R}\cdots\sum_{v_n\in\R}
\prod_{k=1}^{n}
\left[
\chi\bigl(\eps(u_kv_k)\bigr)
Y_{\widetilde{\wt}_S(v_k)}
\right]
\\
&=
\prod_{k=1}^{n}
\left(
\sum_{v_k\in\R}
\chi\bigl(\eps(u_kv_k)\bigr)
Y_{\widetilde{\wt}_S(v_k)}
\right),
\end{align*}
and 
\begin{align*}
|C|\,B_{C^{\circ}}(Y_0,\ldots,Y_m)
=
\sum_{u\in C}
\prod_{k=1}^{n}
\left(
 \sum_{v_k\in\R}
 \chi\bigl(\eps(u_kv_k)\bigr)
 Y_{\widetilde{\wt}_S(v_k)}
\right).
\end{align*}
Now we consider the inner sum for a fixed $u_k\in\R$. Set
$w=\tau_S(v_k)$, Since $\tau_S:\R\to\F_q^m$ is a bijection, as
$v_k$ runs through $\R$, the vector $w$ runs through all of $\F_q^m$, and 
by the definition of $\widetilde{\wt}_S$,
$\widetilde{\wt}_S(v_k)=\wt(\tau_S(v_k))=\wt(w)$.
Also, by Equation~\eqref{eq:pairingsplit},
$\eps(u_kv_k)=\phi_S(u_k)\cdot\tau_S(v_k)=\phi_S(u_k)\cdot w$.
Therefore, after the change of variable $w=\tau_S(v_k)$,
\[
\sum_{v_k\in\R}
\chi\bigl(\eps(u_kv_k)\bigr)
Y_{\widetilde{\wt}_S(v_k)}
=
\sum_{w\in\F_q^m}
\chi\bigl(\phi_S(u_k)\cdot w\bigr)
Y_{\wt(w)}.
\]
If we group the vectors $w\in\F_q^m$ according to their Hamming
weight we get
\begin{align*}
\sum_{w\in\F_q^m}
\chi\bigl(\phi_S(u_k)\cdot w\bigr)
Y_{\wt(w)}
&=
\sum_{j=0}^{m}
Y_j
\sum_{\substack{w\in\F_q^m\\ \wt(w)=j}}
\chi\bigl(\phi_S(u_k)\cdot w\bigr).
\end{align*}
The classical character sum over vectors of Hamming weight $j$ gives
\cite[Ch.~5]{MS}
$$
\sum_{\substack{w\in\F_q^m\\ \wt(w)=j}}
\chi(z\cdot w)
=
K_j(\wt(z)).
$$
Taking $z=\phi_S(u_k)$ and using
$
\wt(\phi_S(u_k))=\wt_S(u_k),
$ 
we obtain
\[
\sum_{v_k\in\R}
\chi\bigl(\eps(u_kv_k)\bigr)
Y_{\widetilde{\wt}_S(v_k)}
=
\sum_{j=0}^{m}
K_j\bigl(\wt_S(u_k)\bigr)Y_j,
\]
and 
substituting this into the previous product gives
\begin{align*}
|C|\,B_{C^{\circ}}(Y_0,\ldots,Y_m)
=
\sum_{u\in C}
\prod_{k=1}^{n}
\left(
\sum_{j=0}^{m}
K_j\bigl(\wt_S(u_k)\bigr)Y_j
\right).
\end{align*}
By Equation~\eqref{eq:A}, and letting
$ 
X_i=\sum_{j=0}^{m}K_j(i)Y_j$, for $0\le i\le m$,
we have that for each coordinate $u_k$,
\[
X_{\wt_S(u_k)}
=
\sum_{j=0}^{m}
K_j\bigl(\wt_S(u_k)\bigr)Y_j.
\]
Therefore,
\[
|C|\,B_{C^{\circ}}(Y_0,\ldots,Y_m)
=
A_C(X_0,\ldots,X_m),
\]
with the above substitutions, which provides the required
identity.
\end{proof}

Theorem~\ref{thm:mwpart} is a refinement of Theorem~\ref{thm:mw}. 
For example, let $m=2$ and $n=3$. Suppose two codewords have
coordinate Gray-weight patterns $(2,0,0)$ and $(1,1,0)$. Both have
total Gray weight $2$, so the ordinary two-variable enumerator records
both by the same monomial $X^4Y^2$. In the composition enumerator,
however, the first contributes $X_0^2X_2$, whereas the second
contributes $X_0X_1^2$.

\begin{remark}{\em
The identities in Theorem~\ref{thm:mw} are recovered from
Theorem~\ref{thm:mwpart} by the substitutions $X_i\longmapsto X^{m-i}Y^i$ and 
$Y_j\longmapsto X^{m-j}Y^j.$
}
\end{remark} 
\section{Quantum-code constructions}\label{sec:css}
The results of the previous sections provide a natural way to obtain
CSS-type quantum codes from QP codes. The main point is that the
annihilator dual remains within the QP family, so dual-containing
conditions can be checked directly over $\R$. When a duality-preserving
coordinate map is available, the same information can be transferred to
ordinary Euclidean duality over $\F_q$. We give the ring and field versions below. Related annihilator-CSS constructions for index one appear in
\cite{BM,BMcorr,AS}; the statements below use the arbitrary-index
duality developed above. The ring version does not require a Gray map,
whereas the field version uses a coordinatewise duality-preserving map
from Section~\ref{sec:gray}.
\vskip 3pt
Let $C_1,C_2\subseteq\R^n$ be two QP codes. We  use  the
equivalence $C_2^\circ\subseteq C_1$ if and only if
$C_1^\circ\subseteq C_2$. Indeed, applying $(\cdot)^\circ$ reverses
inclusions, and $(C^\circ)^\circ=C$ by Theorem~\ref{thm:basic}(e).
\vskip 3pt
Since $\R$ is a finite commutative Frobenius ring, its character module
$\widehat{\R}=\Hom((\R,+),\mathbb C^\ast)$ is isomorphic to $\R$ as an
$\R$-module. In particular, $\R$ admits a generating additive character
$\chi:\R\longrightarrow\mathbb C^\ast$. We use the following
standard consequence of the generating-character property.

\begin{lemma}\label{lem:generating-character}
Let $\chi$ be a generating additive character of $\R$. If $t\in\R$
satisfies $\chi(rt)=1$ for every $r\in\R$, then $t=0$.
\end{lemma}

\begin{proof}
The condition $\chi(rt)=1$ for every $r\in\R$ gives
$Rt\subseteq\ker\chi$. Since the kernel of a generating character
contains no nonzero ideal of $\R$, we have $Rt=\{0\}$. As $1\in\R$,
it follows that $t=0$.
\end{proof}

For $u=(a|b)$ and $v=(a'|b')$ in $\R^{2n}$, where
$a,b,a',b'\in\R^n$, define the $\R$-valued symplectic form by
$\ip{u}{v}_s=b\cdot a'-b'\cdot a$. If
$\mathcal C\subseteq\R^{2n}$ is an $\R$-submodule, we write
\begin{align*}
\mathcal C^{\perp_s}
&=\{v\in\R^{2n}:\ip{u}{v}_s=0\text{ for every }u\in\mathcal C\}
\end{align*}
for its symplectic orthogonal. We also define its character orthogonal by
\begin{align*}
\mathcal C^{\perp_\chi}
&=\{v\in\R^{2n}:\chi(\ip{u}{v}_s)=1
   \text{ for every }u\in\mathcal C\}.
\end{align*}
The value $1$ appears in the second definition because $\chi$ takes
values in the multiplicative group $\mathbb C^\ast$. This is the
multiplicative-character version of the orthogonality used in
\cite{Nadella}.
\vskip 3pt
For $(a|b)\in\R^{2n}$, its symplectic weight is
$\swt(a|b)=\#\{i:(a_i,b_i)\ne(0,0)\}$. The corresponding equality
between the symplectic and character orthogonals appears in
\cite[Lemma~7]{Nadella}. We include the short proof in our notation.

\begin{lemma}\label{lem:symplectic-character}
Let $\mathcal C\subseteq\R^{2n}$ be an $\R$-submodule. Then
$\mathcal C^{\perp_\chi}=\mathcal C^{\perp_s}$.
\end{lemma}

\begin{proof}
If $v\in\mathcal C^{\perp_s}$, then $\ip{u}{v}_s=0$ for every
$u\in\mathcal C$, and therefore
$\chi(\ip{u}{v}_s)=\chi(0)=1$. Hence
$\mathcal C^{\perp_s}\subseteq\mathcal C^{\perp_\chi}$.

Conversely, let $v\in\mathcal C^{\perp_\chi}$ and fix
$u\in\mathcal C$. Put $t=\ip{u}{v}_s$. Since $\mathcal C$ is an
$\R$-submodule, $ru\in\mathcal C$ for every $r\in\R$. Thus
$1=\chi(\ip{ru}{v}_s)=\chi(r\ip{u}{v}_s)=\chi(rt)$ for every
$r\in\R$. By Lemma~\ref{lem:generating-character}, $t=0$. Thus
$\ip{u}{v}_s=0$. Since
$u\in\mathcal C$ was arbitrary, $v\in\mathcal C^{\perp_s}$.
\end{proof}

We now state  the ring version of the CSS construction in this
notation.

\begin{theorem}\label{thm:cssR}
Let $C_1,C_2\subseteq\R^n$ be QP codes satisfying
$C_2^\circ\subseteq C_1$, and put
$\mathcal C=C_1^\circ\times C_2^\circ\subseteq\R^{2n}$. Then
$\mathcal C^{\perp_s}=\mathcal C^{\perp_\chi}=C_2\times C_1$ and,
in particular, $\mathcal C\subseteq\mathcal C^{\perp_\chi}$.
Consequently, there exists an $((n,K,d))_{|\R|}$ stabilizer code, where
$K=|C_1||C_2|/|\R|^n$.
\vskip 2pt \noindent
If $K>1$, then
\begin{align*}
d=\min\bigl\{\,\swt(a|b):\ &(a|b)\in
(C_2\times C_1)\setminus(C_1^\circ\times C_2^\circ)\,\bigr\}.
\end{align*}
If $K=1$, then
\begin{align*}
d=\min\bigl\{\,\swt(a|b):\ &0\ne(a|b)\in C_2\times C_1\,\bigr\}.
\end{align*}
\end{theorem}

\begin{proof}
First, we compute  the symplectic orthogonal of $\mathcal C$. Let
$(x|y)\in\R^{2n}$. By definition,
$(x|y)\in\mathcal C^{\perp_s}$ if and only if
$b\cdot x-y\cdot a=0$ for every $a\in C_1^\circ$ and
$b\in C_2^\circ$.
\vskip 3pt
Taking $a=0$ gives $b\cdot x=0$ for every $b\in C_2^\circ$. Hence
$x\in(C_2^\circ)^{\perp_{\mathcal R}}$. By
Theorem~\ref{thm:basic}(b) and (e),
$(C_2^\circ)^{\perp_{\mathcal R}}
=(C_2^\circ)^\circ=C_2$, so $x\in C_2$.
Similarly, taking $b=0$ gives $y\cdot a=0$ for every
$a\in C_1^\circ$, and therefore
$y\in(C_1^\circ)^{\perp_{\mathcal R}}
=(C_1^\circ)^\circ=C_1$. Thus
$\mathcal C^{\perp_s}\subseteq C_2\times C_1$.
\vskip 3pt
Conversely, let $(x|y)\in C_2\times C_1$. For
$a\in C_1^\circ$ and $b\in C_2^\circ$,
Theorem~\ref{thm:basic}(b) gives $b\cdot x=0$ and $y\cdot a=0$.
Hence $\ip{(a|b)}{(x|y)}_s=b\cdot x-y\cdot a=0$. Therefore
$C_2\times C_1\subseteq\mathcal C^{\perp_s}$, and we obtain
$\mathcal C^{\perp_s}=C_2\times C_1$. By
Lemma~\ref{lem:symplectic-character},
$\mathcal C^{\perp_\chi}=\mathcal C^{\perp_s}=C_2\times C_1$.
\vskip 3pt
Next, we verify the self-orthogonality condition. The hypothesis
$C_2^\circ\subseteq C_1$ implies
$C_1^\circ\subseteq(C_2^\circ)^\circ=C_2$. Together with the original
inclusion, this gives
$\mathcal C=C_1^\circ\times C_2^\circ
\subseteq C_2\times C_1=\mathcal C^{\perp_\chi}$.
The stabilizer criterion of \cite{Nadella} therefore applies.
\vskip 3pt
It remains to compute $K$. By Theorem~\ref{thm:basic}(d),
$|C_i|\,|C_i^\circ|=|\R|^n$ for $i=1,2$, and hence
$|C_i^\circ|=|\R|^n/|C_i|$. Therefore
\begin{align*}
|\mathcal C|
&=|C_1^\circ|\,|C_2^\circ|
 =\frac{|\R|^{2n}}{|C_1||C_2|}.
\end{align*}

The  dimension of the stabilizer code attached to $\mathcal C$ is
$K=|\R|^{n}/|\mathcal C|$, where
\[
K=\frac{|\R|^{n}}{|\mathcal C|}
=|\R|^{n}\cdot\frac{|C_1||C_2|}{|\R|^{2n}}
=\frac{|C_1||C_2|}{|\R|^{n}}.
\]

Finally, if $K>1$, the stabilizer criterion gives
$d=\min\{\swt(z):z\in
\mathcal C^{\perp_\chi}\setminus\mathcal C\}$. Using
$\mathcal C^{\perp_\chi}=C_2\times C_1$ and
$\mathcal C=C_1^\circ\times C_2^\circ$ gives the stated formula.
If $K=1$, \cite{Nadella}  gives
$d=\min\{\swt(z):0\ne z\in\mathcal C^{\perp_\chi}\}$, which gives
the second formula.
\end{proof}

\begin{corollary}
Let $C\subseteq\R^n$ satisfy $C^\circ\subseteq C$, and put
$K=|C|^2/|\R|^n$. Then there exists an $((n,K,d))_{|\R|}$
stabilizer code. If $K>1$, then $d=\min\bigl\{\,\swt(a|b) : (a|b)\in
(C\times C)\setminus(C^\circ\times C^\circ)\,\bigr\}.$
whereas if $K=1$, then $d=\min\bigl\{\,\swt(a|b) : 0\ne(a|b)\in C\times C\,\bigr\}.$
\end{corollary}

\begin{proof}
Apply Theorem~\ref{thm:cssR} with $C_1=C_2=C$.
\end{proof}

We now move to $\F_q$ using a coordinatewise duality-preserving Gray
map, which converts the annihilator-dual inclusion into the usual
Euclidean-dual inclusion required in the CSS construction. 

\begin{theorem}\label{thm:cssFq}
Assume a coordinatewise duality-preserving Gray map $\Phi$ exists. Let $C_1,C_2\subseteq\R^n$ be
QP codes with $C_2^{\circ}\subseteq C_1$, and set
$D_i=\Phi(C_i)\subseteq\F_q^{mn}$, $k_i=\dim_{\F_q}D_i$. Then
$D_2^{\dpe}\subseteq D_1$ and there is a quantum code with parameters $[[\,mn,\;k_1+k_2-mn,\;d\,]]_q,$ where $d=\min\bigl\{\wt(c):c\in(D_1\setminus D_2^{\dpe})\cup
(D_2\setminus D_1^{\dpe})\bigr\}.$
\end{theorem}

\begin{proof}
Since \(\Phi\) is duality-preserving, we have
\(\Phi(C_i^\circ)=D_i^{\perp_E}\) for \(i=1,2\).
Applying \(\Phi\) to the inclusion \(C_2^\circ\subseteq C_1\), we get
\[
    D_2^{\perp_E}\subseteq D_1.
\]
Thus \(D_1\) and \(D_2\) satisfy the usual CSS condition over
\(\mathbb F_q\).

Let \(k_i=\dim_{\mathbb F_q}D_i\). Since \(D_2\) has length \(mn\),
we have \(\dim D_2^{\perp_E}=mn-k_2\). Hence the dimension of the
resulting quantum code is
\[
    k_1-\dim D_2^{\perp_E}
    =k_1-(mn-k_2)
    =k_1+k_2-mn.
\]

Therefore, by the classical CSS construction, there exists a quantum
code with parameters
\[
    [[mn,k_1+k_2-mn,d]]_q,
\]
where $ d=
    \min\bigl\{
        \operatorname{wt}(c):
        c\in
        (D_1\setminus D_2^{\perp_E})
        \cup
        (D_2\setminus D_1^{\perp_E})
    \bigr\}.$
\end{proof}

\begin{example}{\em
We take
$f(x)=x^{4}+x+1$, which is primitive over $\F_2$, so $\R\cong\F_{16}$ is a
field. QP codes of index $5$ are therefore just $\F_{16}$-linear
codes of length $5$, and since $f$ is irreducible Theorem~\ref{thm:global}
identifies the annihilator dual with the ordinary Euclidean dual over
$\F_{16}$. Take the two-generator code $C=\langle a,b\rangle$ with
\[
\begin{aligned}
a&=\bigl(x+x^{3},\;1,\;1+x+x^{3},\;x^{2},\;x^{2}\bigr),\\
b&=\bigl(1,\;1+x^{2},\;x+x^{3},\;x+x^{2}+x^{3},\;0\bigr).
\end{aligned}
\]
One checks $a\cdot a=a\cdot b=b\cdot b=0$ in $\R$, so $C$ is annihilator
self-orthogonal by Corollary~\ref{thm:selforth}(a); as an $\F_{16}$-code $C$
is a self-orthogonal $[5,2]$ code, so $C^{\circ}$ is a $[5,3]$ code and
$\dim_{\F_2}\Psi(C)=8$, $\dim_{\F_2}\Psi(C^{\circ})=12$.
Theorem~\ref{thm:cssFq} with $C_1=C_2=C^{\circ}$ therefore gives
\[
k=2\cdot12-20=4 ,
\]
and exhaustive enumeration over the $48$ Gray maps
$\Phi=\Psi\cdot(I_5\otimes S)$ with $SS^{\mathsf T}=G$ shows that the
minimum weight of $\Phi(C^{\circ})\setminus\Phi(C)$ equals $4$ for
\emph{every} such $S$. Thus this choice gives a $[[20,4,4]]_2$ binary
stabilizer code. By contrast, for the self-orthogonal code generated by
$\bigl(1+x,\,x+x^{3},\,x+x^{3},\,x^{2}+x^{3},\,1+x+x^{2}+x^{3}\bigr)$ and
$\bigl(1+x^{2},\,x^{2}+x^{3},\,x^{3},\,1+x+x^{3},\,x+x^{3}\bigr)$ the same
computation returns minimum weight $2$ or $3$ depending on $S$, so the
choice of the self-orthogonal code, and not only of the Gray map, affects
the resulting distance.}
\end{example}

\begin{example}\label{ex:nonfield}{\em
We take
\[
f(x)=x^{2}+1=(x+1)^{2}\quad\text{over }\F_2,\qquad
\R=\F_2[x]/\langle x^{2}+1\rangle ,
\]
a chain ring of order $4$ with maximal ideal $\langle 1+x\rangle$. Here
$G=I_2$, so $S=\left(\begin{smallmatrix}0&1\\1&0\end{smallmatrix}\right)$
satisfies $SS^{\mathsf T}=G$ and $\Phi_S$ is duality-preserving.

\smallskip\noindent
\textbf{(a)} For $n=8$ let $B=\langle g_1,g_2,g_3\rangle\le\R^{8}$ with
\[
\begin{aligned}
g_1&=(x,\,x,\,1,\,1,\,1,\,1,\,x,\,x),\\
g_2&=(x,\,0,\,x,\,1,\,1+x,\,x,\,1+x,\,0),\\
g_3&=(0,\,x,\,0,\,1+x,\,x,\,1+x,\,x,\,1).
\end{aligned}
\]
All six products $g_i\cdot g_j$ vanish in $\R$, so $B\subseteq B^{\circ}$ by
Corollary~\ref{thm:selforth}(a). Here $\dim_{\F_2}\Psi(B)=5$ and
$\dim_{\F_2}\Psi(B^{\circ})=11$, and Theorem~\ref{thm:cssFq} with
$C_1=C_2=B^{\circ}$ gives a $[[16,6,4]]_2$ stabilizer code.

\smallskip\noindent
\textbf{(b)} For $n=10$ let $B=\langle g_1,g_2,g_3\rangle\le\R^{10}$ with
\[
\begin{aligned}
g_1&=(1,\,1,\,1,\,1,\,1,\,x,\,x,\,x,\,x,\,1),\\
g_2&=(1,\,x,\,1,\,0,\,x,\,0,\,1+x,\,1+x,\,0,\,0),\\
g_3&=(1+x,\,0,\,x,\,1,\,x,\,1,\,x,\,1,\,1+x,\,0).
\end{aligned}
\]
Again $g_i\cdot g_j=0$ for all $i,j$, $\dim_{\F_2}\Psi(B)=6$,
$\dim_{\F_2}\Psi(B^{\circ})=14$, and one obtains a $[[20,8,4]]_2$
stabilizer code.

\smallskip
Both codes meet the best
known lower bound on the minimum distance for their length and dimension \cite{Grassl}. Note that in (a) one has $\dim_{\F_2}\Psi(B)=5$, which is
\emph{not} a multiple of $m=2$: the module $B$ is not free. Over a field
alphabet every submodule of $\R^n$ is free, so $\dim_{\F_2}\Psi(B)$ would
be a multiple of $m$ and the quantum dimension
$k=mn-2\dim_{\F_2}\Psi(B)$ could only lie in $mn-2m\mathbb Z$. The pair
$(16,6)$ is thus unattainable from a field alphabet with $m=2$ through
this one-code specialization $C_1=C_2$ of Theorem~\ref{thm:cssFq}; it is
reached here because $\R$ is not a field, so that $\R^n$ has non-free
submodules. For comparison, in our random searches over the field
$\F_4=\F_2[x]/\langle x^2+x+1\rangle$ with $n=10$, the largest minimum
distance found for the parameters $[[20,8]]_2$ was $3$.}
\end{example}

\subsection{Steane enlargement over $\R$}
We first recall why an enlargement is useful in the present setting.
Suppose that $B\subseteq\R^n$ satisfies $B\subseteq B^\circ$, and we take
 $C_1=C_2=B^\circ$ in the CSS construction from
Theorem~\ref{thm:cssFq}. Since $\Phi$ is duality-preserving,
$
\Phi(B^\circ)=\Phi(B)^{\perp_E}$.
If $r=\dim_{\F_q}\Phi(B)$, then
$
\dim_{\F_q}\Phi(B^\circ)=mn-r$, and hence the corresponding CSS quantum dimension is
$ 
2(mn-r)-mn=mn-2r$.
In particular, this dimension has the same parity as $mn$.

Steane's enlargement \cite{Steane} gives a way to increase the quantum
dimension by replacing one dual-containing classical code by a larger
nested code. The $q$-ary version was established by Ling, Luo and Xing
\cite{LLX}. We use that result here; the point specific to the present
paper is that the required nested pair over $\F_q$ is obtained from
nested QP codes over $\R$ by means of the annihilator dual and a
duality-preserving Gray map.
\vskip 3pt
Let $B'\subsetneq B\subseteq\R^n$
be QP codes and assume that $B\subseteq B^\circ.$ Since annihilator duality reverses inclusions, $B'\subseteq B$ gives
$B^\circ\subseteq B'^\circ.$ Thus we have the chain
$B'\subseteq B\subseteq B^\circ\subseteq B'^\circ.$
After applying a duality-preserving Gray map, this becomes the nested
pair needed for Steane enlargement.

\begin{proposition}\label{thm:steaneR}
Assume that a coordinatewise duality-preserving Gray map $\Phi$ exists.
Let $B'\subsetneq B\subseteq\R^n$ be QP codes satisfying
$B\subseteq B^\circ$. We denote $C=\Phi(B^\circ)$ and $C'=\Phi(B'^\circ),$
with $k=\dim_{\F_q}C$ and $k'=\dim_{\F_q}C'$.
If $k'\ge k+2$, then the $q$-ary Steane enlargement  gives a quantum code with parameters  $\Bigl[\Bigl[\,mn,\;k+k'-mn,\;d\,\Bigr]\Bigr]_q$ with
\[
d\ge
\min\left\{
d(C),\;
\left\lceil\frac{q+1}{q}\,d(C')\right\rceil
\right\}.
\]
\end{proposition}

\begin{proof}
We only need to verify that the classical codes obtained from $B$ and
$B'$ satisfy the hypotheses of the $q$-ary enlargement theorem.
Since $\Phi$ is duality-preserving, we have $\Phi(B^\circ)=\Phi(B)^{\perp_E}.$ 
By the definition of $C$, this gives
$C^{\perp_E}=\Phi(B).$
Similarly,
$C'^{\perp_E}=\Phi(B').$

Now $B'\subseteq B$ implies
$B^\circ\subseteq B'^\circ.$
Together with the hypothesis $B\subseteq B^\circ$, we obtain $B\subseteq B^\circ\subseteq B'^\circ.$
Applying the injective map $\Phi$ gives $\Phi(B)\subseteq\Phi(B^\circ)\subseteq\Phi(B'^\circ).$
Using the definitions of $C$ and $C'$, this is
$C^{\perp_E}\subseteq C\subseteq C'.$
Thus $C$ is dual-containing and $C'$ is an enlargement of $C$.
The assumption $k'\ge k+2$ is exactly the required dimension
condition. The $q$-ary Steane enlargement of 
\cite{LLX} therefore gives the stated quantum dimension and distance
bound.
\end{proof}

\subsection{Codes meeting the best known bounds}
Table~\ref{tab:bestknown} gives the parameters obtained from
Theorem~\ref{thm:cssFq} and Proposition~\ref{thm:steaneR}. 
In each case, the listed distance reaches the best known lower bound in
\cite{Grassl}. When the lower and upper bounds in \cite{Grassl} are equal,
the code is optimal. The matrices defining the duality-preserving Gray maps for the codes in Table~1 are listed in Table~\ref{tab:structures}, while the corresponding generators of $B$ and, for the Steane constructions, of $B'$, are given in Table~\ref{tab:generators}.
 All computations were carried out in
\textsc{Magma} \cite{Bosma}.

\begin{table}[ht!]
\centering
\small
\begin{tabular}{llllll}
\toprule
Parameters & Construction & $\R$ & $n$ & Ref. \cite{Grassl} & Status\\
\midrule
$[[16,6,4]]_2$   & CSS    & $\F_2[x]/\langle x^2+1\rangle$          & 8  & 4   & optimal\\
$[[20,8,4]]_2$   & CSS    & $\F_2[x]/\langle x^2+1\rangle$          & 10 & 4--5 & meets lower bound\\
$[[21,12,3]]_2$  & Steane & $\F_8=\F_2[x]/\langle x^3+x+1\rangle$   & 7  & 3--4 & meets lower bound\\
$[[24,4,6]]_2$   & Steane & $\F_2[x]/\langle (x^2+x+1)^2\rangle$    & 6  & 6--8 & meets lower bound\\
\midrule
$[[12,5,3]]_3$   & Steane & $\F_3[x]/\langle (x-1)^3\rangle$        & 4  & 3--4 & meets lower bound\\
$[[12,6,3]]_3$   & Steane & $\F_3[x]/\langle (x-1)^3\rangle$        & 4  & 3   & optimal\\
$[[14,8,3]]_3$   & Steane & $\F_3[x]/\langle x^2-1\rangle$          & 7  & 3   & optimal\\
$[[15,8,3]]_3$   & Steane & $\F_3[x]/\langle (x-1)^3\rangle$        & 5  & 3--4 & meets lower bound\\
$[[16,10,3]]_3$  & Steane & $\F_3[x]/\langle x^2-1\rangle$          & 8  & 3   & optimal\\
$[[18,12,3]]_3$  & Steane & $\F_3[x]/\langle x^2-1\rangle$          & 9  & 3   & optimal\\
$[[20,14,3]]_3$  & Steane & $\F_9=\F_3[x]/\langle x^2+2x+2\rangle$  & 10 & 3   & optimal\\
\bottomrule
\end{tabular}
\caption{Quantum codes obtained from QP codes meeting the best-known
lower bound in \cite{Grassl}.}
\label{tab:bestknown}
\end{table}

For $[[12,5,3]]_3$ and $[[15,8,3]]_3$, we have $m=3$. Over a field,
the dimensions occurring in the CSS and Steane constructions are
multiples of $3$, and hence the resulting quantum dimension is also a
multiple of $3$. Therefore, the dimensions $5$ and $8$ cannot occur in
the field case. These two examples are possible because $\R$ is not a field and the
corresponding codes $B$ are non-free $\R$-modules with
$\dim_{\F_q}\Psi(B)=5.$

In each case, we write $B=\langle g_1,\dots,g_r\rangle$, and for the
Steane constructions we take $B'=\langle g_1,\dots,g_{r-1}\rangle$.
The generators are chosen so that $g_i\cdot g_j=0$ for all $i,j$,
and hence $B\subseteq B^\circ$.
For the Gray map, we use $\Phi_S$ with
$SS^{\mathsf T}=\lambda G$. The choices of $S$ used for the codes in
Table~\ref{tab:bestknown} were checked in \textsc{Magma} \cite{Bosma}, and they are in Table~\ref{tab:structures}. 

\begin{table}[ht!]
\centering
\footnotesize
\begin{tabular}{cccc}
\toprule
Ring & $G$ & $S$ & $\lambda$ \\
\midrule
$\mathbb{F}_2[x]/\langle x^2+1\rangle$, $\mathbb{F}_3[x]/\langle x^2-1\rangle$, and $\mathbb{F}_9$ 
& $I_m$ 
& Coordinate-reversing permutation matrix 
& $1$ \\
\addlinespace
$\mathbb{F}_8$  (over $\mathbb{F}_2$)
& $\begin{pmatrix} 1 & 0 & 0 \\ 0 & 0 & 1 \\ 0 & 1 & 0 \end{pmatrix}$
& $\begin{pmatrix} 1 & 1 & 1 \\ 0 & 1 & 1 \\ 1 & 0 & 1 \end{pmatrix}$
& $1$ \\
\addlinespace
 $\mathbb{F}_3[x]/\langle x^3-1\rangle$
& $\begin{pmatrix} 1 & 0 & 0 \\ 0 & 0 & 1 \\ 0 & 1 & 0 \end{pmatrix}$
& $\begin{pmatrix} 0 & 1 & 1 \\ 1 & 1 & 2 \\ 1 & 2 & 1 \end{pmatrix}$
& $2$ \\
\addlinespace
$\mathbb{F}_2[x]/\langle x^4+x^2+1\rangle$
& $\begin{pmatrix} 1 & 0 & 0 & 0 \\ 0 & 0 & 0 & 1 \\ 0 & 0 & 1 & 0 \\ 0 & 1 & 0 & 1 \end{pmatrix}$
& $\begin{pmatrix} 0 & 0 & 0 & 1 \\ 0 & 1 & 1 & 0 \\ 1 & 0 & 0 & 0 \\ 0 & 0 & 1 & 0 \end{pmatrix}$
& $1$ \\
\bottomrule
\end{tabular}
\caption{Choices for the codes in Table~\ref{tab:bestknown}.}
\label{tab:structures}
\end{table}

 \begin{table}[ht!]
\centering
\footnotesize
\begin{tabular}{lccl}
\toprule
Parameters & $\mathcal{R}$ & $n$ & Generators \\
\midrule
$[[16,6,4]]_2$ & $\mathbb{F}_2[x]/\langle x^2+1\rangle$ & 8 & $
\begin{aligned}[t]
g_1&=(x,\,x,\,1,\,1,\,1,\,1,\,x,\,x),\\
g_2&=(x,\,0,\,x,\,1,\,1+x,\,x,\,1+x,\,0),\\
g_3&=(0,\,x,\,0,\,1+x,\,x,\,1+x,\,x,\,1).  
\end{aligned} Example~\ref{ex:nonfield}
$ \\
\addlinespace
$[[20,8,4]]_2$ & $\mathbb{F}_2[x]/\langle x^2+1\rangle$ & 10 & $
\begin{aligned}[t]
g_1&=(1,\,1,\,1,\,1,\,1,\,x,\,x,\,x,\,x,\,1),\\
g_2&=(1,\,x,\,1,\,0,\,x,\,0,\,1+x,\,1+x,\,0,\,0),\\
g_3&=(1+x,\,0,\,x,\,1,\,x,\,1,\,x,\,1,\,1+x,\,0).
\end{aligned}
$ Example~\ref{ex:nonfield} \\
\addlinespace
$[[21,12,3]]_2$ & $\mathbb{F}_2[x]/\langle x^3+x+1\rangle$ & 7 & 
$\begin{aligned}[t]
g_1 &= (x+1,\ x^2+x,\ 1,\ x^2+1,\ x^2+x+1,\ x^2+1,\ x+1) \\
g_2 &= (x^2+1,\ x+1,\ 0,\ x+1,\ x^2+1,\ x,\ x)
\end{aligned}$ \\
\addlinespace
$[[24,4,6]]_2$ & $\mathbb{F}_2[x]/\langle x^4+x^2+1\rangle$ & 6 & 
$\begin{aligned}[t]
g_1 &= (x,\ x^2+x,\ x^3+x^2+x+1,\ 0,\ x^3+x^2,\ x^3+1) \\
g_2 &= (x^3+x^2+x+1,\ x^3+x+1,\ x,\ x^2,\ x^3,\ x^2) \\
g_3 &= (x^3+x^2,\ x^3+x^2+x+1,\ x+1,\ x^3+x^2+x,\ x^2,\ x+1)
\end{aligned}$ \\
\addlinespace
$[[12,5,3]]_3$ & $\mathbb{F}_3[x]/\langle x^3-1\rangle$ & 4 & 
$\begin{aligned}[t]
g_1 &= (2x^2+x,\ 2x^2+2x+2,\ x^2+2x,\ x+2) \\
g_2 &= (2x^2+x+2,\ x^2+x+2,\ 2x^2,\ x^2+2)
\end{aligned}$ \\
\addlinespace
$[[12,6,3]]_3$ & $\mathbb{F}_3[x]/\langle x^3-1\rangle$ & 4 & 
$\begin{aligned}[t]
g_1 &= (2x+1,\ x^2+2x,\ 2x^2+2x+2,\ 2x^2+x) \\
g_2 &= (0,\ 2x^2+x,\ x^2+2x,\ 2x+1)
\end{aligned}$ \\
\addlinespace
$[[14,8,3]]_3$ & $\mathbb{F}_3[x]/\langle x^2-1\rangle$ & 7 & 
$\begin{aligned}[t]
g_1 &= (1,\ 2x+2,\ x,\ 1,\ x,\ 2x,\ 2x+1) \\
g_2 &= (2x+2,\ 1,\ 2,\ x+2,\ x,\ 2,\ x)
\end{aligned}$ \\
\addlinespace
$[[15,8,3]]_3$ & $\mathbb{F}_3[x]/\langle x^3-1\rangle$ & 5 & 
$\begin{aligned}[t]
g_1 &= (2x+1,\ x^2+x+1,\ x+2,\ x^2+x+1,\ 2x^2+1) \\
g_2 &= (x^2,\ 2x^2,\ x^2+2,\ x^2+2,\ 2x+2)
\end{aligned}$ \\
\addlinespace
$[[16,10,3]]_3$ & $\mathbb{F}_3[x]/\langle x^2-1\rangle$ & 8 & 
$\begin{aligned}[t]
g_1 &= (x+2,\ 2x+2,\ 2x,\ 2x,\ x+1,\ x,\ 2,\ x+2) \\
g_2 &= (2x,\ x+2,\ 1,\ 0,\ 1,\ 1,\ x+1,\ x)
\end{aligned}$ \\
\addlinespace
$[[18,12,3]]_3$ & $\mathbb{F}_3[x]/\langle x^2-1\rangle$ & 9 & 
$\begin{aligned}[t]
g_1 &= (2x+2,\ x+1,\ x,\ 1,\ 2,\ x,\ 2x,\ 1,\ x+1) \\
g_2 &= (x,\ 2x+1,\ 2x,\ 2x+2,\ 2x+2,\ 0,\ 2x+1,\ x,\ 1)
\end{aligned}$ \\
\addlinespace
$[[20,14,3]]_3$ & $\mathbb{F}_9 = \mathbb{F}_3[x]/\langle x^2+2x+2\rangle$ & 10 & 
$\begin{aligned}[t]
g_1 &= (x+1,\ 2x,\ 1,\ x+2,\ x,\ x+1,\ 1,\ x+2,\ x+2,\ x) \\
g_2 &= (0,\ x+2,\ 2x,\ x,\ x+1,\ x+2,\ 0,\ 2x+2,\ 1,\ 1)
\end{aligned}$ \\
\bottomrule
\end{tabular}
\caption{Generators for various quantum codes.}
\label{tab:generators}
\end{table}

\section{Conclusion}\label{sec:conclusion}

In this paper, we studied annihilator duality for QP codes over
$\R=\F_q[x]/\langle f\rangle$ with $f(0)\ne0$. We showed that
$C^\circ=C^{\perp_{\mathcal R}}$, so the annihilator dual of a QP code
is again QP. We also described this duality in coefficient coordinates
through a symmetric Hankel Gram matrix and obtained the corresponding
twisted Euclidean description. For squarefree $f$, the Chinese Remainder decomposition reduces
annihilator duality to Euclidean duality on the constituents. This gives
criteria for self-orthogonal, self-dual, dual-containing and LCD QP
codes. We also obtained the MacWilliams
transform for the Gray image and its refinement to Gray-weight
compositions. We applied the duality results to CSS and
Steane-enlarged quantum-code constructions over $\R$ and, when a
coordinatewise duality-preserving Gray map exists, over $\F_q$. The
codes obtained include eleven binary and ternary stabilizer codes whose
 minimum distance meets the best known lower bound, six of
them optimal. 

\section*{Acknowledgment}
E. Martínez-Moro is partially supported by Grant {\small SAGACT-1
MCIN/AEI/10.13039/501100011033} y FEDER ``Una manera de hacer Europa''
PID2022-138906NB-C21 (2023/2027). Also financial support of the Department of Education of the Junta de Castilla y León and FEDER Funds is gratefully acknowledged (Reference: CLU-2025-1-02- IMUVA).
 D. Panario was funded by the Natural Sciences and Engineering Research Council of Canada (NSERC), reference number RPGIN-2024-05341.


\begin{thebibliography}{99}

\bibitem{ADLS}
A.~Alahmadi, S.~T.~Dougherty, A.~Leroy, P.~Sol\'e,
On the duality and the direction of polycyclic codes,
\emph{Adv. Math. Commun.} \textbf{10}(4) (2016), 921--929.

\bibitem{BMS}
M.~Bajalan, E.~Mart\'inez-Moro, S.~Szabo,
A transform approach to polycyclic and serial codes over rings,
\emph{Finite Fields Appl.} \textbf{80} (2022), 102014.


\bibitem{BMSSY}
M.~Bajalan, E.~Mart\'inez-Moro, R.~Sobhani, S.~Szabo,
G.~G.~Y{\i}lmazg\"u\c{c},
On the structure of repeated-root polycyclic codes over local rings,
\emph{Discrete Math.} \textbf{347}(1) (2024), 113715.




\bibitem{BM}
M.~Bajalan, E.~Mart\'inez-Moro,
Polycyclic codes over serial rings and their annihilator CSS
construction,
\emph{Cryptogr. Commun.} \textbf{17}(2) (2025), 283--306.

\bibitem{BMcorr}
M.~Bajalan, E.~Mart\'inez-Moro,
Correction to: Polycyclic codes over serial rings and their annihilator
CSS construction,
\emph{Cryptogr. Commun.} \textbf{17}(3) (2025), 661--664.

\bibitem{BP}
T.~Bag, D.~Panario,
Quasi-polycyclic and skew quasi-polycyclic codes over $\F_q$,
\emph{Finite Fields Appl.} \textbf{101} (2025), 102536.

\bibitem{BEA}
T.~P.~Berger, N.~El~Amrani,
Codes over finite quotients of polynomial rings,
\emph{Finite Fields Appl.} \textbf{25} (2014), 165--181.

\bibitem{Bosma}
W. Bosma, J. Cannon, and C. Playoust,
\newblock The {M}agma algebra system. {I}. {T}he user language,
\newblock \emph{J. Symbolic Comput.} \textbf{24} (1997), no. 3--4, 235--265,
\newblock Computational algebra and number theory (London, 1993).

\bibitem{CRSS}
A.~R.~Calderbank, E.~M.~Rains, P.~W.~Shor, N.~J.~A.~Sloane,
Quantum error correction via codes over $GF(4)$,
\emph{IEEE Trans. Inf. Theory} \textbf{44}(4) (1998), 1369--1387.

\bibitem{FMB}
A.~Fotue-Tabue, E.~Mart\'inez-Moro, J.~T.~Blackford,
On polycyclic codes over a finite chain ring,
\emph{Adv. Math. Commun.} \textbf{14}(3) (2020), 455--466.

\bibitem{Grassl}
M.~Grassl, Bounds on the minimum distance of linear codes and quantum
codes, online available at \url{http://www.codetables.de}, accessed  on 12 September 2026.

\bibitem{HL}
T.~Honold, I.~Landjev,
MacWilliams identities for linear codes over finite Frobenius rings,
in: D.~Jungnickel, H.~Niederreiter (Eds.),
\emph{Finite Fields and Applications (Augsburg, 1999)},
Springer, Berlin--Heidelberg, 2001, pp.~276--292.

\bibitem{KKKS}
A.~Ketkar, A.~Klappenecker, S.~Kumar, P.~K.~Sarvepalli,
Nonbinary stabilizer codes over finite fields,
\emph{IEEE Trans. Inf. Theory} \textbf{52}(11) (2006), 4892--4914.

\bibitem{LLX}
S.~Ling, J.~Luo, C.~Xing,
Generalization of Steane's enlargement construction of quantum codes and
applications,
\emph{IEEE Trans. Inf. Theory} \textbf{56}(8) (2010), 4080--4084.

\bibitem{LPS}
S.~R.~L\'opez-Permouth, B.~R.~Parra-Avila, S.~Szabo,
Dual generalizations of the concept of cyclicity of codes,
\emph{Adv. Math. Commun.} \textbf{3}(3) (2009), 227--234.

\bibitem{MS}
F.~J.~MacWilliams, N.~J.~A.~Sloane,
\emph{The Theory of Error-Correcting Codes},
North-Holland, Amsterdam, 1977.

\bibitem{Nadella}
S.~Nadella, A.~Klappenecker,
Stabilizer codes over Frobenius rings,
in \emph{Proc. IEEE Int. Symp. Inf. Theory} (2012), 165--169.

\bibitem{ONA}
H.~Ou-azzou, M.~Najmeddine, N.~Aydin,
On the algebraic structure of quasi-polycyclic codes and new quantum
codes,
\emph{Quantum Inf. Process.} \textbf{23} (2024), Article 95.


\bibitem{SL}
G.~Seroussi, A.~Lempel,
Factorization of symmetric matrices and trace-orthogonal bases in finite
fields,
\emph{SIAM J. Comput.} \textbf{9}(4) (1980), 758--767.

\bibitem{Steane}
A.~M.~Steane,
Enlargement of Calderbank--Shor--Steane quantum codes,
\emph{IEEE Trans. Inf. Theory} \textbf{45}(7) (1999), 2492--2495.

\bibitem{SPD}
K.~Suxena, O.~Prakash, I.~Debnath, P.~Sol\'e,
Optimal GQPC codes over the finite field $\F_q$,
\emph{Mathematics} \textbf{13}(22) (2025), 3655.

\bibitem{SU}
S.~Szabo, F.~Ulmer,
Duality preserving Gray maps for codes over rings,
\emph{J. Algebra Appl.} \textbf{16}(9) (2017), 1750161.


\bibitem{AS}
A.~Tiwari, R.~Sarma,
Polycyclic codes over the $\F_q$-algebra $\F_q^{\ell}$ and their
annihilator dual,
\emph{Comput. Appl. Math.} \textbf{44}(7) (2025), Article 358.

\bibitem{WSS}
R.~Wu, M.~Shi, P.~Sol\'e,
On the structure of $1$-generator quasi-polycyclic codes over finite
chain rings,
\emph{J. Appl. Math. Comput.} \textbf{68}(5) (2022), 3491--3503.

\end{thebibliography}
\end{document}